\documentclass[10pt,a4paper]{article}

\usepackage[T1]{fontenc}
\usepackage[utf8]{inputenc} 
\usepackage{lmodern}

\usepackage[a4paper,margin=2.5cm]{geometry}
\usepackage{microtype}

\usepackage{amsmath,amssymb,amsfonts,amsthm,mathtools}

\usepackage{graphicx}
\usepackage{booktabs}
\usepackage{multirow}
\usepackage{array}
\usepackage{tabularx}
\usepackage{caption}
\usepackage{subcaption}

\usepackage{enumitem}

\usepackage{algorithm}
\usepackage{algpseudocode}

\usepackage{siunitx}
\usepackage{xspace}
\usepackage{url}

\usepackage{graphicx}
\usepackage{booktabs}
\usepackage{multirow}
\usepackage{makecell}
\usepackage{diagbox}
\usepackage[table]{xcolor}
\usepackage{caption}

\usepackage[colorlinks=true,linkcolor=blue,citecolor=blue,urlcolor=blue]{hyperref}

\usepackage[nameinlink,capitalize,noabbrev]{cleveref}

\newtheorem{theorem}{Theorem}

\newtheorem{lemma}{Lemma}
\theoremstyle{definition}
\newtheorem{definition}{Definition}

\usepackage{orcidlink}

\usepackage[nameinlink,capitalize,noabbrev]{cleveref}
\usepackage{titlesec}

\crefname{appendix}{appendix}{appendices}
\Crefname{appendix}{Appendix}{Appendices}

\newcommand{\keywords}[1]{%
  \par\smallskip\noindent\textbf{Keywords: }#1
}

\begin{document}

\title{Polyhedral Unmixing: Bridging Semantic Segmentation with Hyperspectral Unmixing via Polyhedral-Cone Partitioning\\[0.25em]}

\author{
Antoine Bottenmuller\,\orcidlink{0009-0007-3943-3419}, 
Etienne Decenci\`ere\,\orcidlink{0000-0002-1349-8042} and 
Petr Dokladal\,\orcidlink{0000-0002-6502-7461}
\\[1.0em]
\textit{Mines Paris, PSL University, Centre for Mathematical Morphology (CMM)}\\
\textit{77300 Fontainebleau, France}\\[0.5em]
\texttt{\normalsize\{antoine.bottenmuller, etienne.decenciere, petr.dokladal\}@minesparis.psl.eu}\\[1.5em]
{\normalfont Code: \url{https://github.com/antoine-bottenmuller/polyhedral-unmixing}}
}

\date{}

\maketitle

\begin{abstract}

Semantic segmentation and hyperspectral unmixing are two central problems in spectral image analysis. The former assigns each pixel a discrete label corresponding to its material class, whereas the latter estimates pure material spectra, called endmembers, and, for each pixel, a vector representing material abundances in the observed scene. Despite their complementarity, these two problems are usually addressed independently. This paper aims to bridge these two lines of work by formally showing that, under the linear mixing model, pixel classification by dominant materials induces polyhedral-cone regions in the spectral space. We leverage this fundamental property to propose a direct segmentation-to-unmixing pipeline that performs blind hyperspectral unmixing from any semantic segmentation by constructing a polyhedral-cone partition of the space that best fits the labeled pixels. Signed distances from pixels to the estimated regions are then computed, linearly transformed via a change of basis in the distance space, and projected onto the probability simplex, yielding an initial abundance estimate. This estimate is used to extract endmembers and recover final abundances via matrix pseudo-inversion. Because the segmentation method can be freely chosen, the user gains explicit control over the unmixing process, while the rest of the pipeline remains essentially deterministic and lightweight. Beyond improving interpretability, experiments on three real datasets demonstrate the effectiveness of the proposed approach when associated with appropriate clustering algorithms, and show consistent improvements over recent deep and non-deep state-of-the-art methods. 

\end{abstract}


\section{Introduction}
\label{sec:intro}

\begin{figure*}[t]
    \centering
    \includegraphics[width=1.0\linewidth]{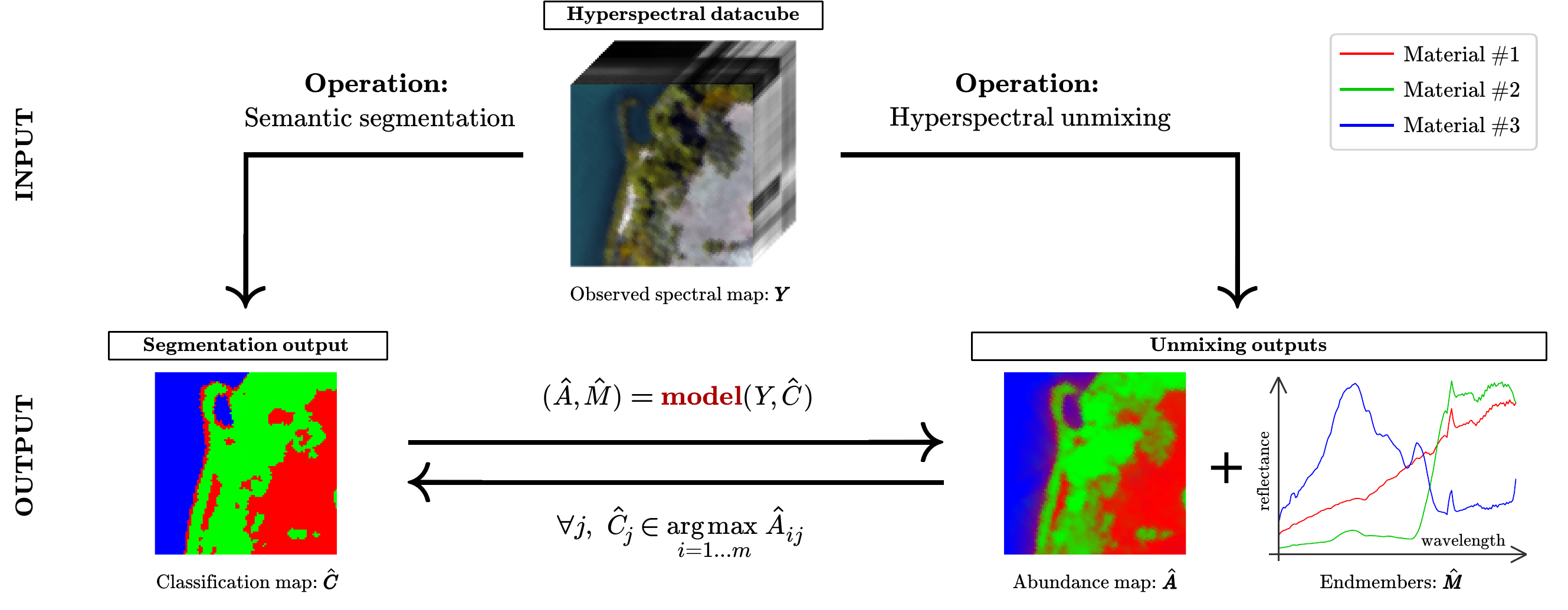}
    \caption{Bridge between semantic segmentation and hyperspectral unmixing. From the unmixing data, the dominant-material classification map is directly obtained by taking the $\arg\max$ over the abundances for each pixel. In the opposite way, a model is needed.}
    \label{fig:segm_vs_unmix}
\end{figure*}

In remote sensing, hyperspectral imaging systems acquire information about an object or a scene without physical contact~\cite{bhargava2024hyperspectral,khan2018modern}. They capture discretized spectral signatures of materials, typically representing their absorbance or reflectance, over tens to hundreds of contiguous, narrow bands~\cite{landgrebe2002hyperspectral}. 
A hyperspectral dataset is therefore a particular type of image where, for $d$ bands, each pixel is an element in the Euclidean space $\mathbb{R}^d$, referred to as the \textit{spectral space}~\cite{bioucas2013hyperspectral}, where the $i$-th entry is the spectral response at the $i$-th band. 

Because of limited spatial resolution or the intrinsic heterogeneity of the observed medium, the spectral signature of a single pixel results from a mixture of multiple
materials~\cite{rasti2024image,bioucas2012hyperspectral}. 
The spectral signatures of pure materials are referred to as \textit{endmembers}, and their proportions within a pixel are \textit{abundances}~\cite{rasti2024image,wei2020overview,bioucas2012hyperspectral}. 
The relationship between observed signatures, endmembers, and abundances is defined by a \textit{mixing model}. We focus on the standard widely-used \textit{linear mixing model} (LMM)~\cite{keshava2002spectral,mendoza2024blind}, which neglects nonlinear interactions by modelling each pixel as a linear combination of endmembers weighted by their abundances.

Given a dataset $Y$, hyperspectral image analysis typically addresses two central problems~\cite{bioucas2013hyperspectral,plaza2009recent}: 

\textbf{1. Semantic segmentation} (or \textit{Classification}~\cite{bioucas2013hyperspectral}). It assigns each pixel of $Y$ a label representing its material class, yielding a \textit{classification map} $C$ \cite{kang2018building,soucy2023ceu}.

\textbf{2. Hyperspectral unmixing}. It estimates pure material signatures (endmembers $M$) and their pixelwise fractions (abundances $A$) \cite{duan2023undat,wei2020overview,wei2016multiband}.

Research on semantic segmentation for hyperspectral image classification spans supervised and unsupervised, deep and non-deep, techniques~\cite{lv2023deep,pan2020dssnet,trajanovski2020tongue}. A particular case is \textit{dominant-material} segmentation (or classification)~\cite{charles2011learning,fang2018unsupervised}, where the assigned class corresponds to the material with maximal abundance in each mixed pixel. We adopt this setting in this paper. Hyperspectral unmixing is more demanding, as it seeks both $M$ and $A$. It has been approached via geometrical, statistical, nonnegative matrix factorization, and deep-learning methods~\cite{bioucas2012hyperspectral,wei2020overview,bhatt2020deep}. Recent works include autoencoder-based models and deep unrolling of classical algorithms, frequently achieving stronger performance~\cite{rasti2024image}.

We place ourselves in the \textit{blind} setting, where no prior information about pure materials is provided, and only unsupervised techniques are allowed~\cite{rasti2023sunaa,rasti2024image}. Optionally, only the number of materials is known.
In hyperspectral semantic segmentation, the methods often leverage clustering or classical unsupervised techniques performed in the spectral space (e.g., $k$-means, Gaussian mixture models), with spatial priors~\cite{soucy2023ceu}. 
Classical blind linear unmixing follows a two-step pipeline~\cite{zhong2016blind,wei2020overview}: (i) endmember extraction, then (ii) abundance estimation. 
Dominant-material labels follow by taking the $\arg\max$ of the abundances in each pixel~\cite{ibarrola2019hyperspectral}. A preprocessing step is also often applied~\cite{keshava2002spectral}. Some statistical and deep learning-based approaches jointly estimate $M$ and $A$~\cite{bioucas2012hyperspectral,rasti2024image}.

Despite their complementarity, these two problems are usually addressed independently~\cite{keshava2002spectral,mendoza2024blind}. While dominant-material segmentation $\hat{C}$ is directly derived from unmixing outputs $\hat{A}$ and $\hat{M}$~\cite{ibarrola2019hyperspectral}, the reverse direction, i.e., performing unmixing from $Y$ and $\hat{C}$, requires a \textit{model} that, to our knowledge, has not been formulated yet. 
\Cref{fig:segm_vs_unmix} illustrates the link between these two problems. 

\textbf{Contributions.} We propose a segmentation-to-unmixing model through a new blind linear unmixing approach by rethinking the standard pipeline, bridging semantic segmentation with hyperspectral unmixing. We prove that, under the LMM, dominant-material classification induces a partition of the spectral space into polyhedral cones. 
Leveraging this property, and given any semantic segmentation, we fit a polyhedral-cone partition to the labeled pixels
; compute signed distances from pixels to the corresponding convex regions; 
and project onto the probability simplex to obtain an initial abundance estimate. This estimate is then used to recover endmembers and final abundances. 

This reversed-pipeline strategy makes segmentation drive unmixing, giving users explicit control through the choice of segmentation method, while the subsequent steps are lightweight and reproducible. 
Beyond improving interpretability, the approach opens new paths to hyperspectral unmixing. Experiments on three widely used real datasets, namely Samson, Urban and Jasper Ridge, show consistent improvements over recent deep and non-deep baselines.

Contrarily to existing blind unmixing methods incorporating segmentation information, which use clustering to select \textit{candidate pure pixels} and then apply standard endmember extraction~\cite{martin2011region,martin2012spatial,kowkabi2016fast,xu2018regional}, our model derives a direct segmentation-to-unmixing pipeline based on a fundamental polyhedral-based geometric property of dominant-material regions under the LMM that we formalize and prove in this paper. While our method is compatible with supervised segmentation methods, we focus here on unsupervised pixel classification (clustering) algorithms to perform blind unmixing.\\

The remainder of the paper is organized as follows. \Cref{sec:related} recalls the fundamentals of blind unmixing under the LMM, reviews benchmark methods and highlights the standard processing chain. \Cref{sec:method} establishes the polyhedral-cone partitioning property and presents our segmentation-driven unmixing algorithm. \Cref{sec:experiments} reports experimental results and their comparisons to eight baseline methods. \Cref{sec:conclusion} concludes and outlines future directions.

\section{Related Work}
\label{sec:related}

This section aims to give background knowledge about blind hyperspectral unmixing under linear mixing assumption, and fixes notations used throughout the paper.

\subsection{Linear Mixing Model}

The LMM models the observed data $Y$ as linear combinations of endmembers $M$ weighted by their abundances $A$. If an image contains $n$ pixels, each represented by a spectrum of $d$ bands, and is assumed to involve $m$ material endmembers, the LMM is then formulated by the matrix equation
\begin{equation}
\label{eq:yma}
    Y = MA + E ,
\end{equation}
\noindent where $Y \in \mathbb{R}^{d \times n}$ is the matrix stacking the $n$ pixel spectra; $M \in \mathbb{R}^{d \times m}$ is the matrix of the $m$ endmember spectra; $A \in \mathbb{R}^{m \times n}$ is the abundance matrix (every real $A_{ij}$ is the proportion of the $i$-th endmember in the $j$-th pixel); and $E \in \mathbb{R}^{d \times n}$ is a residual zero-mean additive Gaussian noise with low variance~\cite{duan2023undat}. $A$ satisfies, for all $j=1,\ldots,n$,
\begin{equation}
\label{eq:abundance_constraints}
    \left\{
    \begin{array}{l}
        A_{ij} \geq 0 \text{ for } i=1,\ldots,m\\
        \sum_{i=1}^m A_{ij} = 1 
    \end{array}
    \right. .
\end{equation}
Constraints \eqref{eq:abundance_constraints} ensure that the data $Y$, when neglecting noise $E$ in \eqref{eq:yma}, lie within the $(m-1)$-simplex defined by the endmembers $M$ in the spectral space~\cite{wei2020overview,bioucas2012hyperspectral}.

%
%
%
\subsection{Existing Approach Overview}

In the \textit{blind} scenario, the objective of hyperspectral unmixing is to recover both the matrices $M$ and $A$ in~\eqref{eq:yma}, whose estimations are denoted by $\hat{M}$ and $\hat{A}$, respectively, using only the observed data $Y$~\cite{rasti2023sunaa,rasti2024image}. This constitutes a blind source separation (BSS) problem~\cite{zhong2016blind,xia2011independent,bofill2001underdetermined}. 
Numerous algorithms have been developed to address the blind linear unmixing problem, which, based on three comprehensive surveys~\cite{rasti2024image,wei2020overview,bioucas2012hyperspectral}, can be broadly grouped into four main categories.

\textbf{Geometrical approaches} focus on endmember extraction by exploiting the geometry of the data cloud $Y$ and the endmember simplex to determine $M$~\cite{bioucas2012hyperspectral}. They fall into two main groups: (i) \textit{minimum volume} (MV) algorithms, which fit the smallest-volume enclosing simplex, such as the minimum volume simplex analysis (MVSA)~\cite{li2015minimum} or the simplex identification via split augmented Lagrangian (SISAL)~\cite{huang2022sisal}; and (ii) \textit{pure pixel} (PP) algorithms, which assume endmembers are present among the pixels and identify them directly, such as the vertex component analysis (VCA)~\cite{nascimento2005vertex}.

\textbf{Nonnegative matrix factorization (NMF)} finds $\hat{M}$ and $\hat{A}$ by minimizing the squared Frobenius norm $\|Y-\hat{M}\hat{A}\|_F^2$ with an additional regularization term on $\hat{M}$, under nonnegativity constraints and abundance constraints~\eqref{eq:abundance_constraints}~\cite{feng2022hyperspectral}. Notable variants include the minimum volume constrained-nonnegative matrix factorization (MVC-NMF)~\cite{miao2007endmember}, which incorporates a minimum-volume penalty, and the iterative constrained endmembers (ICE) algorithm~\cite{berman2004ice}, which regularizes using the cumulative distance between endmembers.

\textbf{Statistical approaches} incorporate prior knowledge to regularize the solution space and perform maximum a posteriori (MAP) inference, typically under Gaussian likelihoods~\cite{bioucas2012hyperspectral,dobigeon2009joint}. Under Gaussian noise assumption, this leads to constrained optimization problems similar in form to NMF~\cite{wei2020overview}. Examples include sparse Bayesian learning~\cite{chen2014hyperspectral}, and dependent component analysis (DECA)~\cite{nascimento2011hyperspectral} which assumes a mixture of Dirichlet densities as a prior for the abundances.

\textbf{Deep learning methods} 
often recast classical unmixing models into deep neural architectures~\cite{bhatt2020deep,duan2023undat}. 
Autoencoder variants map abundances to the latent space (the code) and endmembers to the linear decoder weights, with training typically minimizing a regularized Frobenius loss similar in form to NMF~\cite{palsson2018hyperspectral,palsson2022blind}; popular examples include EndNet~\cite{ozkan2018endnet} and the minimum simplex convolutional network (MiSiCNet)~\cite{rasti2022misicnet}. Unrolling approaches embed iterative solvers (e.g., the alternating direction method of multipliers (ADMM) or NMF) into trainable networks, such as ADMMNet~\cite{zhou2021admm} or SNMF-Net~\cite{xiong2021snmf}.

\begin{figure*}[t]
    \centering
    \includegraphics[width=1.0\linewidth]{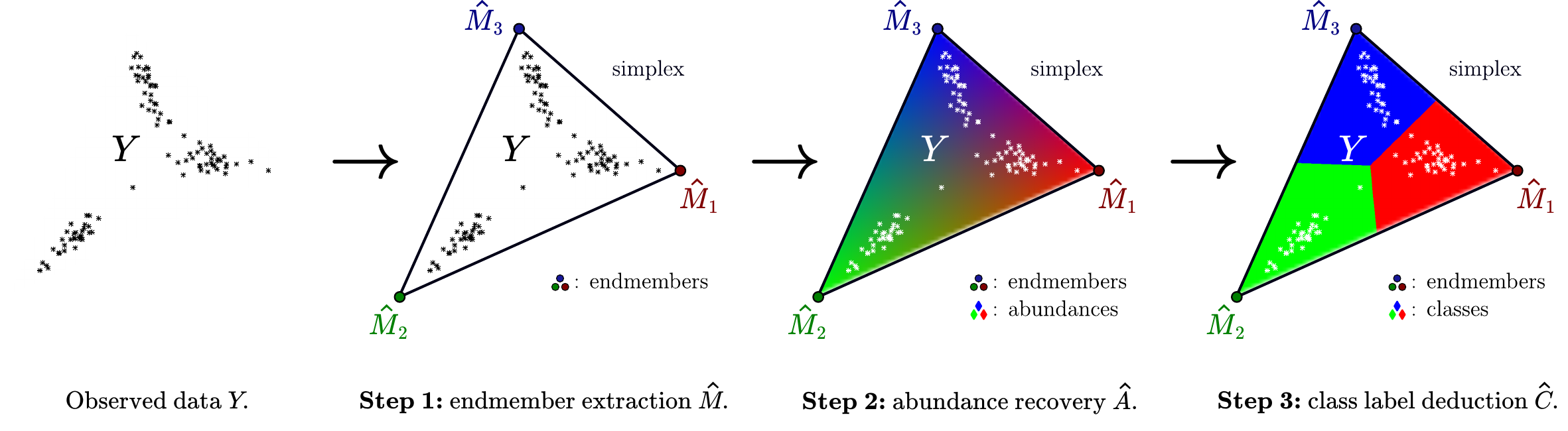}
    \caption{Illustration, in the spectral space, of the standard pipeline for blind linear umixing. 2D cross-sections of a 3D space.}
    \label{fig:standard_chain}
\end{figure*}

\subsection{Standard Processing Chain}

Given a number $m$ of endmembers, set a priori or determined using a specific method such as \textit{virtual dimensionality}~\cite{chang2004estimation}, most blind linear unmixing methods follow a three-step processing chain~\cite{zhong2016blind,rasti2024image,wei2020overview}, illustrated in \Cref{fig:standard_chain}:

\begin{enumerate}
    \item \textbf{Endmember extraction}: the endmember matrix $\hat{M}$ is estimated from the observed spectral data $Y$ alone, using a chosen unsupervised endmember extractor~\cite{bioucas2012hyperspectral,rasti2024image} such as the ones discussed in the previous subsection.
    
    \item \textbf{Abundance recovery}: sometimes performed jointly with endmember extraction, this \textit{inversion step}~\cite{eches2011variational,keshava2002spectral} estimates the abundance matrix $\hat{A}$ from both $Y$ and $\hat{M}$. This is done either 
    by singular value decomposition (SVD) matrix inversion~\cite{van2012collinearity} 
    or via least squares or regression-based methods~\cite{wei2015fast,xiong2021snmf}.
    
    \item \textbf{(Optional) Data labeling}: each pixel can be assigned a dominant-material label by maximum abundance classification~\cite{ibarrola2019hyperspectral}. The class $\hat{C}_j$ of pixel $j$ is given by 
    \begin{equation}
        \hat{C}_j \in \arg\max_{i = 1,\ldots,m} \hat{A}_{ij} .
    \end{equation}
\end{enumerate}

Pre-processing steps such as luminance correction and dimensionality reduction are usually applied to the original dataset $Y$~\cite{bioucas2013hyperspectral,keshava2002spectral}, 
by projecting the data onto a lower-dimensional subspace, typically via principal component analysis (PCA)~\cite{bioucas2012hyperspectral,rasti2024image}.

Outlining the traditional unmixing chain provides a useful framework for understanding our proposed method, which inverts the standard sequence. 

\section{Proposed Unmixing Method}
\label{sec:method}

We state in this section our main theorem on the geometric structure of the spectral space under the LMM, and derive from it a new linear hyperspectral unmixing method driven by semantic segmentation in the blind setting. 
The proofs of \cref{lem:convexity,lem:partition} and of \cref{them:main} are provided in \Cref{app:proofs}.

\subsection{The Polyhedral-Cone Partition Theorem}

We work in the Euclidean space $\mathbb{R}^d$ modeling the spectral space, and begin with the definitions of the following fundamental geometric objects. 

\begin{definition}[Halfspace]
\label{def:halfspace}
A subset $H$ of $\mathbb{R}^d$ is a (closed) halfspace if there exist $w \in \mathbb{R}^{d*}$ and $b \in \mathbb{R}$ such that
\begin{equation}
\label{eq:halfspace}
    H = \left\{ x \in \mathbb{R}^d ~\middle|~ \langle x , w \rangle \leq b \right\} .
\end{equation}
It is exactly one of the two closed subsets of $\mathbb{R}^d$ resulting in the split of the space by an affine hyperplane~{\normalfont{\cite{matouvsek2007understanding}}}. 
\end{definition}

\begin{definition}[Polyhedron]
\label{def:polyhedron}
A (convex) polyhedron, or \textit{polyhedral set}, of $\mathbb{R}^d$ is the intersection of finitely-many closed halfspaces of $\mathbb{R}^d$~{\normalfont{\cite{bruns2009polytopes}}}.
\end{definition}

\begin{definition}[Polyhedral Cone]
\label{def:cone}
A \textit{polyhedral cone} of $\mathbb{R}^d$ is a polyhedron whose boundary hyperplanes all pass through the origin $0_d$ (i.e., $b=0$)~{\normalfont{\cite{barker1981theory,amelunxen2017intrinsic}}}.
\end{definition}

Polyhedral sets, as intersections of convex sets, are then also convex. A polyhedral cone is a particular polyhedron. We next formalize the notion of \textit{dominant-material regions}.

\begin{definition}[Dominant-Material Region]
\label{def:region}
Given a class $c \in \{1,\ldots,m\}$, its dominant-material region $\mathcal{R}_c$ of $\mathbb{R}^d$ is
\begin{equation}
\label{eq:region}
    \mathcal{R}_c := \left\{ x \in \mathbb{R}^d ~\middle|~ c \in \arg\max_{i = 1, \ldots, m} a_i \right\} ,
\end{equation}
where, for each $x \in \mathbb{R}^d$, the $a_i \in \mathbb{R}$ are the material abundances associated with $x$ over the endmembers $M_{:i} \in \mathbb{R}^d$.
\end{definition}

Assuming the endmembers $\{M_{:i}\}_{i=1}^m$ are linearly independent, under the LMM, the abundances $a_i$ are exactly the unique linear coefficients $\lambda_i$ of $x$ over the $M_{:i}$, where
\begin{equation}
\label{eq:linear_combination}
    x = \sum_{i=1}^m \lambda_i M_{:i} + y, \quad y \in (\text{span}\{M_{:i}\}_{i=1}^m)^\perp.
\end{equation}

Each region is thus defined by the $\arg\max$ over the $\lambda_i$, which is preserved under convex combinations. This yields our first lemma:

\begin{lemma}[Region convexity]
\label{lem:convexity}
Under linear mixture assumption, dominant-material regions $\mathcal{R}_c$ of $\mathbb{R}^d$ are convex.
\end{lemma}

By the Hyperplane Separation Theorem (th. 7.3 in~\cite{gallier2011geometric}), any two nonempty disjoint convex regions of a Euclidean space can be separated by a hyperplane. This yields our second lemma (see~\cite{leon2018spaces}):

\begin{lemma}[Polyhedral Property of Convex Partitions]
\label{lem:partition}
Any convex finite ($m$-)partition of a Euclidean space results in $m$ convex polyhedral regions.
\end{lemma}

Up to a set of Lebesgue measure zero, namely where the $\arg\max$ in \eqref{eq:region} is not unique, the dominant-material regions form an $m$-partition of $\mathbb{R}^d$, which, by \cref{lem:convexity}, is convex. By \cref{lem:partition}, these regions are polyhedral. 
Moreover, the origin $0_d$ has $0$ for coefficient on every endmember and thus belongs to every dominant-material region, hence to all separation hyperplanes. We finally obtain:

\begin{theorem}[Polyhedral-Cone Partition Theorem]
\label{them:main}
Under the LMM, the dominant-material regions are $m$ polyhedral cones forming an $m$-partition of $\mathbb{R}^d$.
\end{theorem}

A 2D cross-section of such a partition in $\mathbb{R}^3$ is illustrated in \Cref{fig:standard_chain} (step 3). This theorem motivates us to start unmixing with the identification of a polyhedral-cone partition of the spectral space that best separates the data labeled by semantic segmentation. We develop our method hereinafter.

\subsection{Method Overview}

Our method inverts the standard pipeline (\cref{fig:standard_chain}), with segmentation as entry point and \Cref{them:main} as the geometric backbone. 
It comprises three main stages:
\begin{enumerate}
    \item \textbf{Polyhedral-cone partitioning:} the pixels are classified by a clustering or unsupervised segmentation method, yielding a classification map $\hat{C}$, and a polyhedral-cone partition of the spectral space is fitted over the labeled data.
    \item \textbf{Initial abundance estimate:} a first abundance estimate $\hat{A}'$ is computed from the space partition using the signed distances from the data to the polyhedral-cone regions.
    \item \textbf{Endmember and abundance recovery:} from $\hat{A}$, endmembers $\hat{M}$ are extracted and final abundances $\hat{A}$ are recovered by matrix pseudo-inversion.
\end{enumerate}
These stages are detailed in the next three subsections.

\subsection{Polyhedral-Cone Partitioning}

\begin{figure*}[t]
    \centering
    \includegraphics[width=1.00\linewidth]{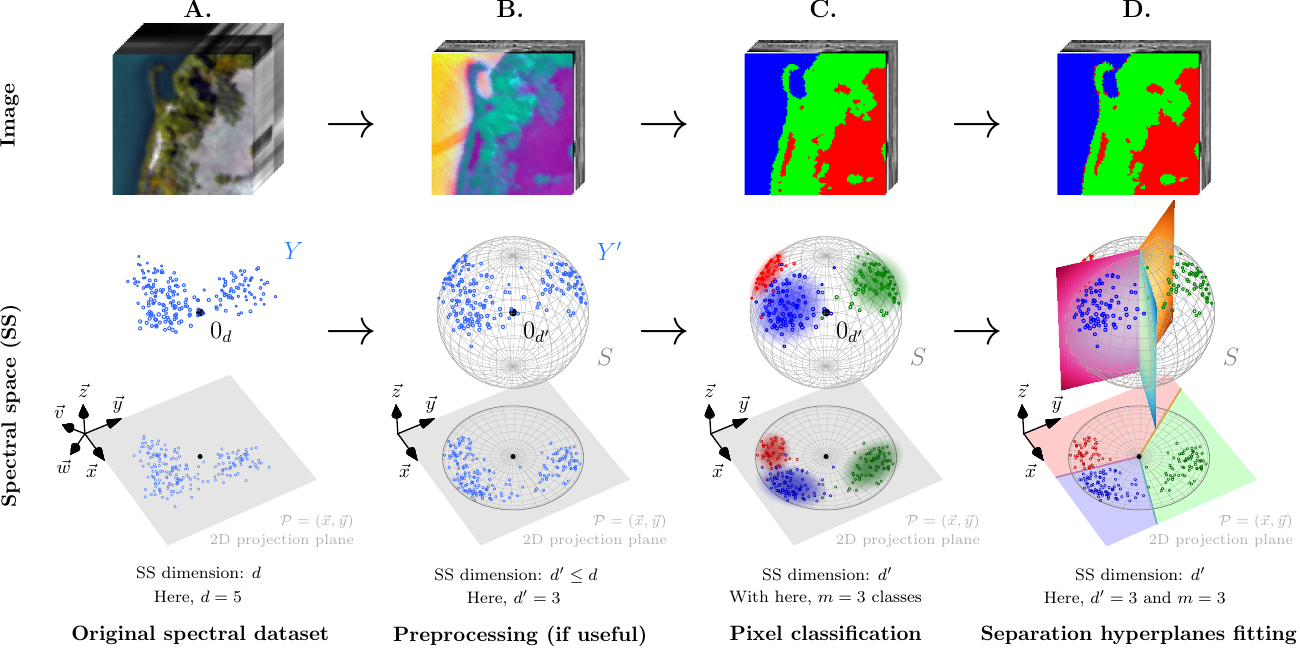}
    \caption{Processing chain to determine a polyhedral-cone partition of the spectral space. From left to right: observed data (A.) are pre-processed (B.) and classified into $m$ classes (C.). Separation hyperplanes are then computed (D.), inducing $k \geq m$ polyhedral regions. Top row: \textit{image} domain; bottom row: \textit{spectral space} representation (3D view + 2D projection plane $P(\vec{x}, \vec{y})$ below, which allows a better visualization of the data and the three polyhedral regions). Each RGB color is associated with one class.}
    \label{fig:partition}
\end{figure*}

The first stage aims to find a polyhedral-cone partition of the spectral space that best separates the data labeled by any chosen (unsupervised) semantic segmentation method. 
It is composed of three sub-steps, illustrated in \Cref{fig:partition}:
\begin{enumerate}
    \item (A. $\rightarrow$ B.) \textbf{Preprocessing}. A preliminary step optionally maps $Y$ to $Y'$. We may (i) normalize spectral luminance by projecting the data onto the unit sphere $S$ of $\mathbb{R}^{d}$ to mitigate shadows and illumination effects~\cite{zouaoui2023entropic}, and (ii) reduce space dimensionality to $d' < d$ via PCA~\cite{keshava2002spectral,rasti2024image}.
    \item (B. $\rightarrow$ C.) \textbf{Semantic segmentation}. A hyperspectral semantic segmentation method (via pixel classification or clustering)~\cite{li2025hyperspectral} is applied to $Y'$ to obtain a classification map $\hat{C}$ into $m$ classes, where $m$ is either fixed or estimated~\cite{chang2004estimation,bioucas2012hyperspectral,mclachlan2014number}. For instance, Gaussian Mixture Models (GMM)~\cite{rasmussen1999infinite,yang2012robust} provide a flexible clustering tool for non-homogenous spectral data~\cite{li2013hyperspectral}.
    \item (C. $\rightarrow$ D.) \textbf{Separation hyperplanes}. Leveraging \Cref{them:main} and the Hyperplane Separation Theorem~\cite{rudinfunctional}, linear hyperplanes that best pairwise separate the $m$ class clusters are estimated. This can be done via an unbiased linear support vector machine (SVM)~\cite{pedregosa2011scikit}.
\end{enumerate}

There are $\binom{m}{2}=\frac{1}{2}m(m-1)$ pairwise separation hyperplanes, which induce a partition of the spectral space into $k \geq m$ nonempty polyhedral cones, where $k$ depends on the arrangement of hyperplanes. In practice, we retain only the $m$ largest regions that contain the most data. They serve as approximations of the true dominant-material regions. 

\subsection{Initial Abundance Estimate}

\begin{figure*}[t]
    \centering
    \includegraphics[width=1.00\linewidth]{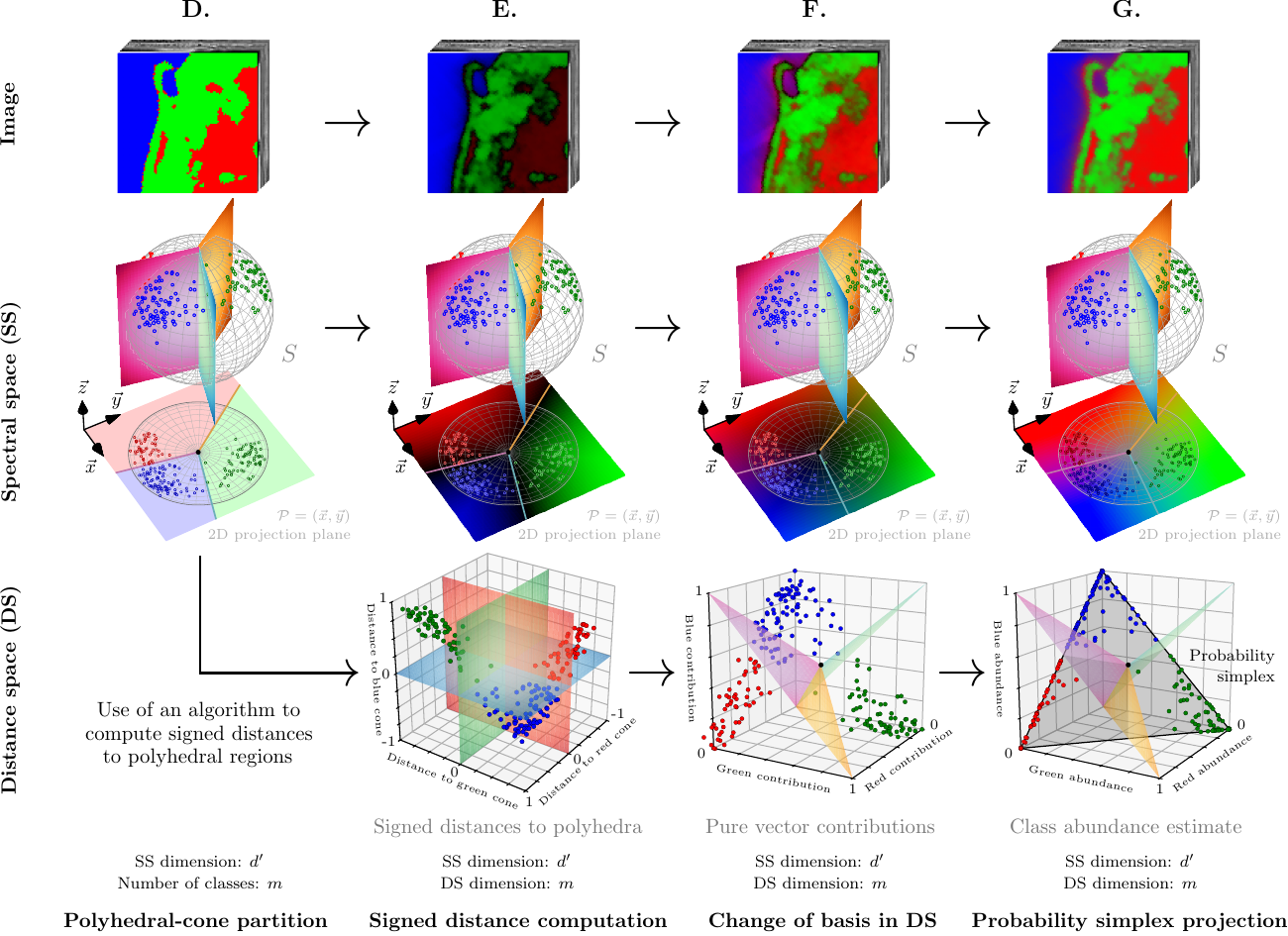}
    \caption{Processing chain for computing an abundance estimate. From left to right: given the spectral partition (D.), signed distances to each cone are computed (E.); a change of basis is then applied in the distance space (F.), and the resulting vectors are projected onto the probability simplex (G.). Top: \textit{image}; middle: \textit{spectral space}; bottom: \textit{distance space}. The 2D RGB projection plane $P(\vec{x}, \vec{y})$ shows: polyhedral regions (D.), negative distances in each region (E.), transformed distance vectors (F.) and abundance vectors (G.). Each RGB color is associated with one class. Yellow, magenta and cyan represent separation planes between R-G, R-B and G-B pairs of classes, respectively.}
    \label{fig:abundance}
\end{figure*}

The second stage computes an initial abundance estimate $\hat{A}'$ from the partition. It consists in three sub-steps (see \cref{fig:abundance}):
\begin{enumerate}
    \item (D. $\rightarrow$ E.) \textbf{Signed distance computation}. For each point $y_i' \in \mathbb{R}^{d'}$, we compute its signed distance to every class region $\mathcal{R}_c$, obtaining an $m$-dimensional distance vector. For $x \in \mathbb{R}^{d'}$ and $S \subseteq \mathbb{R}^{d'}$, the signed distance $d_s$ is
    \begin{equation}
        d_s(x,S) := \left\{\begin{array}{cl}
            d(x,S) & \text{ if } x \notin S \\
            - d(x,S^\complement) & \text{ otherwise}
        \end{array}\right. ,
    \end{equation}
    where $d(x,S) := \inf_{y \in S}\|x-y\|$ is the Euclidean distance. 
    When $S$ is a polyhedron $\mathcal{R}$, $d(x,\mathcal{R})$ is obtained by solving a nearest-point problem in a polyhedral set~\cite{liu2012nearest,yu2025projections}. An exact algorithm for the signed distance $d_s$ is given in~\cite{bottenmuller2025euclidean}. 
    Stacking these vectors yields a distance matrix $D \in \mathbb{R}^{m \times n}$ in the \textit{distance space} $\mathbb{R}^m$. Each coordinate stands for one cone (see \cref{fig:abundance}, E.).

    \item (E. $\rightarrow$ F.) \textbf{Change of basis}. We then perform a change of basis in the distance space so that the new coordinates reflect contributions with respect to a set of \textit{reference distance vectors}. They can be chosen, for instance, as the vectors of smallest signed distances in each class (they are the closest to the regions core). 
    Let $B \in \mathbb{R}^{m \times m}$ stack these $m$ reference vectors as columns. We transform $D' := B^{-1} D$. 
    This mitigates scale disparities between regions (some cones may be much larger, yielding larger raw distances) and aligns the transformed vectors with the probability simplex, preparing for projection.
    
    \item (F. $\rightarrow$ G.) \textbf{Simplex projection}. Finally, we project the columns of $D'$ onto the probability simplex~\cite{patterson2013stochastic,wang2013projection} in $\mathbb{R}^m$. Since the simplex is a polyhedral set, this projection can again be solved via nearest-point algorithms such as~\cite{bottenmuller2025euclidean}. Just before projection, we scale $D'$ by a factor $s > 0$, which controls the \textit{saturation} of the projected data: larger $s$ pushes the points closer to the simplex edges. This is the only hyperparameter of our method.
\end{enumerate}
This stage therefore outputs an initial abundance estimate $\hat{A}' \in \mathbb{R}^{m \times n}$ by projecting the columns of $D'$ onto the probability simplex in the distance space.

\subsection{Endmember and Abundance Recovery}

The final stage recovers from $\hat{A}'$ and $Y$ both endmembers $\hat{M}$ and refined abundances $\hat{A}$ via matrix pseudo-inversion, in a supervised unmixing fashion~\cite{rasti2024image}. 
It contains two successive sub-steps:
\begin{enumerate}
    \item \textbf{Endmember recovery}. The endemember matrix $\hat{M}$ is first estimated from $\hat{A}'$ by $\hat{M} := Y \hat{A}'^+_\lambda$, where $\hat{A}'^+_\lambda := \hat{A}'^\top (\hat{A}' \hat{A}'^\top + \lambda I_m)^{-1}$ is the regularized Moore–Penrose pseudo-inverse of $\hat{A}'$, with $\lambda$ the Tikhonov regularization parameter~\cite{day2011inversion}. $\hat{M}$ is the analytic solution to $\min_M \|Y - M\hat{A}'\|_F^2 + \lambda \|M\|_F^2$.
    
    \item \textbf{Abundance recovery}. The final abundance matrix $\hat{A}$ is then estimated from $\hat{M}$ by $\hat{A} := \hat{M}^+_\lambda Y$, where $\hat{M}^+_\lambda := (\hat{M}^\top \hat{M} + \lambda I_m)^{-1} \hat{M}^\top$ is the regularized Moore–Penrose pseudo-inverse of $\hat{M}$. $\hat{A}$ is the analytic solution to $\min_A \|Y - \hat{M}A\|_F^2 + \lambda \|A\|_F^2$.
\end{enumerate}
The regularization parameter $\lambda \geq 0$ is set to 0 when the matrices are well-conditioned. 
Note that the considered data are the initial observations $Y$, so that $\hat{M}$ is estimated in the original spectral space $\mathbb{R}^d$. 
The pair $(\hat{M},\hat{A})$ constitutes the final output of our proposed unmixing method.

\section{Experiments}
\label{sec:experiments}

We evaluate our unmixing method in the blind setting on three real datasets, and compare the results with eight state-of-the-art unmixing algorithms. The visualizations of the results (estimated abundance maps and endmembers for unmixing experiments, and classification maps for robustness evaluation to different clustering methods) are provided in \Cref{app:visualizations}.

\subsection{Hyperspectral Datasets}

We consider three widely-used real hyperspectral datasets, all publicly available with their ground truths~\cite{zouaoui2023entropic,zhu2017hyperspectral}:
\begin{itemize}
    \item \textbf{Samson}~\cite{zhu2014spectral}. The Samson dataset is a 95$\times$95 pixels sub-image of a larger dataset. It uses a total of 156 spectral bands spanning from 401 to 889 nm. Three main materials are observed in the scene: \#1 Soil, \#2 Tree and \#3 Water.
    
    \item \textbf{Jasper Ridge}~\cite{zhu2014structured}. Jasper Ridge is a 100$\times$100 pixels sub-image of a larger dataset captured by the AVIRIS sensor~\cite{green1998imaging}. It uses 224 spectral bands spanning from 380 to 2500 nm. 26 noisy bands are removed before processing, leaving a total of 198 bands. Four materials are observed in the scene: \#1 Soil, \#2 Tree, \#3 Water and \#4 Road.
    
    \item \textbf{Urban-6}~\cite{qian2011hyperspectral}. The Urban-6 dataset is a 307$\times$307 pixels image captured by the HYDICE sensor~\cite{rickard1993hydice}. It uses 210 spectral bands spanning from 400 to 2500 nm. 48 of them are removed as a pre-processing step due to atmospheric noise effects, leaving a total of 162 bands. The dataset contains six materials: \#1 Asphalt, \#2 Grass, \#3 Tree, \#4 Roof, \#5 Metal and \#6 Dirt.
\end{itemize}

\subsection{Benchmark Methods} 

We compare to eight strong baselines across three families:
\begin{itemize}
    \item \textbf{Geometrical}: one pure-pixel (PP) algorithm, VCA~\cite{nascimento2005vertex}; and one minimum-volume (MV) algorithm, SISAL~\cite{bioucas2009variable}, among the most effective~\cite{rasti2024image}.
    \item \textbf{NMF}: a quadratic minimum volume one, NMF-QMV~\cite{zhuang2019regularization}, which is a strong benchmark NMF algorithm~\cite{rasti2024image}.
    \item \textbf{Deep learning-based}: three autoencoders, (i) Endnet~\cite{ozkan2018endnet}, which includes a loss function with several terms, (ii) MiSiCNet~\cite{rasti2022misicnet}, which incorporates geometrical information, and (iii) a recent cascaded spatial-spectral Mamba-based model (UNMamba)~\cite{chen2025unmamba}; and one deep unrolling model, ADMMNet~\cite{zhou2021admm}, which embed the ADMM~\cite{boyd2011distributed} solver into a trainable architecture. We also test Entropic Descent Archetypal Analysis (EDAA)~\cite{zouaoui2023entropic} which uses an entropic descent algorithm to solve the Archetypal Analysis problem~\cite{rasti2024image}.
\end{itemize}
MiSiCNet has been shown to outperform prior deep and non-deep methods~\cite{rasti2022misicnet,zouaoui2023entropic}, such as CyCUNet~\cite{gao2021cycu}, uDAS~\cite{qu2018udas} and UnDIP~\cite{rasti2021undip}. EDAA outperforms MiSiCNet~\cite{zouaoui2023entropic}.
ADMMNet improves over uDAS, UnDIP, EGU-Net-pw~\cite{hong2021endmember} and MNN-BU-2~\cite{qian2020spectral}. 
These make these methods competitive baselines.

\subsection{Evaluation Metrics}

To evaluate the unmixing quality, we consider two standard metrics~\cite{zhu2014spectral,ozkan2018endnet}: (i) the spectral angle distance (SAD) for endmembers, and (ii) the root mean square error (RMSE) for abundances. 
Let $i \in \{1,\ldots,m\}$ denote the index of the material among the observed ones to evaluate.

The SAD between the predicted endmember $\hat{M}_{:i} \in \mathbb{R}^d$ and the associated ground-truth $M_{:i} \in \mathbb{R}^d$ is defined as
\begin{equation}
    \text{SAD}(\hat{M}_{:i},M_{:i}) = \arccos{\left(\frac{\langle \hat{M}_{:i} , M_{:i} \rangle}{\|\hat{M}_{:i}\| \|M_{:i}\|}\right)} .
\end{equation}
It measures the angle (radians) between two vectors in $\mathbb{R}^d$. 

The RMSE between the predicted abundance $\hat{A}_{i:} \in \mathbb{R}^n$ and its associated ground-truth $A_{i:} \in \mathbb{R}^n$ is defined as
\begin{equation}
    \text{RMSE}(\hat{A}_{i:}, A_{i:}) = \sqrt{\frac{1}{n} \sum_{j=1}^n (\hat{A}_{ij} - A_{ij})^2} .
\end{equation}
It quantifies the mean per-pixel error. For both metrics, lower is better. We also report averages (Avg.) over materials for every experiment and every dataset.


\definecolor{colorEX1}{rgb}{0.40, 0.40, 0.80}
\definecolor{colorEX2}{rgb}{0.40, 0.00, 0.00}

\begin{table*}[t]
\caption{The SADs and RMSEs (mean$\pm$std over 10 runs per experiment) on the Samson dataset. The best results are shown in bold; the second best are underlined.}
\label{tab:results_samson}
  \centering
  \small
  \setlength{\tabcolsep}{8pt}
  \renewcommand{\arraystretch}{1.2}

  {\scriptsize
  \resizebox{\textwidth}{!}{%
  \begin{tabular}{|
      p{1.25cm}||
      >{\centering\arraybackslash}p{1.5cm}
      >{\centering\arraybackslash}p{1.5cm}
      >{\centering\arraybackslash}p{1.7cm}
      >{\centering\arraybackslash}p{1.6cm}
      >{\centering\arraybackslash}p{1.6cm}
      >{\centering\arraybackslash}p{1.6cm}
      >{\centering\arraybackslash}p{1.6cm}
      >{\centering\arraybackslash}p{1.6cm}
      >{\centering\arraybackslash}p{1.7cm}|
    }
    \hline
    \multirow{2}{*}{Endm.} &
    \multicolumn{9}{c|}{\textbf{Endmember Spectral Angle Distance (SAD)} $(\times 10^{-2})$} \\
    \cline{2-10}
    & \textbf{VCA} & \textbf{SISAL} & \textbf{NMF-QMV} & \textbf{Endnet} & \textbf{MiSiCNet} & \textbf{ADMMNet} & \textbf{EDAA} & \textbf{UNMamba} & \textbf{Proposed} \\
    \hline
    \#1 Soil  & 13.26 $\pm$ 9.3 & 32.05 $\pm$ 0.0 &  8.91 $\pm$ 0.0 & \textbf{1.42 $\pm$ 0.3} & 1.91 $\pm$ 0.2 &  2.24 $\pm$ 0.0 & 2.87 $\pm$ 0.0 & 6.13 $\pm$ 0.9 & \underline{1.85 $\pm$ 0.2} \\
    \#2 Tree  &  4.39 $\pm$ 2.7 &  \textbf{2.96 $\pm$ 0.0} &  8.44 $\pm$ 0.0 & 3.48 $\pm$ 0.3 & 5.76 $\pm$ 0.0 &  5.32 $\pm$ 0.0 & 3.45 $\pm$ 0.0 & 4.48 $\pm$ 0.4 & \underline{3.41 $\pm$ 0.1} \\
    \#3 Water & 11.96 $\pm$ 0.7 &  6.89 $\pm$ 0.0 & 18.87 $\pm$ 0.0 & 4.06 $\pm$ 0.3 & 9.45 $\pm$ 0.1 & 12.48 $\pm$ 0.0 & \textbf{2.28 $\pm$ 0.0} & 3.38 $\pm$ 0.1 & \underline{2.93 $\pm$ 0.2} \\
    \hline
    Avg. & 9.87 $\pm$ 4.1 & 13.97 $\pm$ 0.0 & 12.07 $\pm$ 0.0 & 2.99 $\pm$ 0.2 & 5.71 $\pm$ 0.1 & 6.68 $\pm$ 0.0 & \underline{2.87 $\pm$ 0.0} & 4.66 $\pm$ 0.2 & \textbf{2.73 $\pm$ 0.1} \\
    \hline
  \end{tabular}}}
  
  \,

  {\scriptsize
  \resizebox{\textwidth}{!}{%
  \begin{tabular}{|
      p{1.25cm}||
      >{\centering\arraybackslash}p{1.5cm}
      >{\centering\arraybackslash}p{1.5cm}
      >{\centering\arraybackslash}p{1.7cm}
      >{\centering\arraybackslash}p{1.6cm}
      >{\centering\arraybackslash}p{1.6cm}
      >{\centering\arraybackslash}p{1.6cm}
      >{\centering\arraybackslash}p{1.6cm}
      >{\centering\arraybackslash}p{1.6cm}
      >{\centering\arraybackslash}p{1.7cm}|
    }
    \hline
    \multirow{2}{*}{Endm.} &
    \multicolumn{9}{c|}{\textbf{Abundance Root Mean Square Error (RMSE)} $(\times 10^{-2})$} \\
    \cline{2-10}
    & \textbf{VCA} & \textbf{SISAL} & \textbf{NMF-QMV} & \textbf{Endnet} & \textbf{MiSiCNet} & \textbf{ADMMNet} & \textbf{EDAA} & \textbf{UNMamba} & \textbf{Proposed} \\
    \hline
    \#1 Soil  & 18.14 $\pm$ 9.2 & 16.26 $\pm$ 0.0 & 13.67 $\pm$ 0.0 & 15.99 $\pm$ 0.6 & 6.76 $\pm$ 0.3 & 10.04 $\pm$ 0.0 & \underline{5.73 $\pm$ 0.0} & 12.55 $\pm$ 0.6 & \textbf{4.18 $\pm$ 0.3} \\
    \#2 Tree  & 15.45 $\pm$ 8.7 &  7.34 $\pm$ 0.0 &  8.40 $\pm$ 0.0 & 20.22 $\pm$ 0.4 & 5.42 $\pm$ 0.3 &  5.88 $\pm$ 0.0 & \underline{3.76 $\pm$ 0.0} & 10.55 $\pm$ 0.7 & \textbf{3.05 $\pm$ 0.2} \\
    \#3 Water & 10.20 $\pm$ 5.7 & 14.89 $\pm$ 0.0 & 11.61 $\pm$ 0.0 & 27.36 $\pm$ 0.4 & 3.59 $\pm$ 0.1 &  8.10 $\pm$ 0.0 & \textbf{2.59 $\pm$ 0.0} &  \underline{2.61 $\pm$ 0.1} & 3.33 $\pm$ 0.2 \\
    \hline
    Avg. & 14.59 $\pm$ 7.9 & 12.83 $\pm$ 0.0 & 11.23 $\pm$ 0.0 & 21.19 $\pm$ 0.4 & 5.26 $\pm$ 0.2 & 8.01 $\pm$ 0.0 & \underline{4.03 $\pm$ 0.0} &  8.57 $\pm$ 0.4 & \textbf{3.52 $\pm$ 0.2} \\
    \hline
    \hline
    Time (s) & 0.03 & 0.13 & 6.38 & 328.51 & 104.09 & 46.75 & 19.72 & 20.65 & \textcolor{colorEX1}{0.38} $+$ \textcolor{colorEX2}{0.06} \\
    \hline
  \end{tabular}}}
\end{table*}

\begin{table*}[t]
\caption{The SADs and RMSEs (mean$\pm$std over 10 runs per experiment) on the Jasper Ridge dataset. The best results are shown in bold; the second best are underlined.}
\label{tab:results_jasper_ridge}
  \centering
  \small
  \setlength{\tabcolsep}{8pt}
  \renewcommand{\arraystretch}{1.2}

  {\scriptsize
  \resizebox{\textwidth}{!}{%
  \begin{tabular}{|
      p{1.25cm}||
      >{\centering\arraybackslash}p{1.5cm}
      >{\centering\arraybackslash}p{1.5cm}
      >{\centering\arraybackslash}p{1.7cm}
      >{\centering\arraybackslash}p{1.6cm}
      >{\centering\arraybackslash}p{1.6cm}
      >{\centering\arraybackslash}p{1.6cm}
      >{\centering\arraybackslash}p{1.6cm}
      >{\centering\arraybackslash}p{1.6cm}
      >{\centering\arraybackslash}p{1.7cm}|
    }
    \hline
    \multirow{2}{*}{Endm.} &
    \multicolumn{9}{c|}{\textbf{Endmember Spectral Angle Distance (SAD)} $(\times 10^{-2})$} \\
    \cline{2-10}
    & \textbf{VCA} & \textbf{SISAL} & \textbf{NMF-QMV} & \textbf{Endnet} & \textbf{MiSiCNet} & \textbf{ADMMNet} & \textbf{EDAA} & \textbf{UNMamba} & \textbf{Proposed} \\
    \hline
    \#1 Soil  & 37.39 $\pm$ 9.6 & 16.34 $\pm$ 0.0 & 21.65 $\pm$ 0.0 & 21.54 $\pm$ 2.7 &  7.21 $\pm$ 0.1 & 76.62 $\pm$ 0.0 & \underline{5.71 $\pm$ 0.0} & 14.68 $\pm$ 0.1 & \textbf{1.97 $\pm$ 0.1} \\
    \#2 Tree  & 23.52 $\pm$ 4.1 &  3.64 $\pm$ 0.0 & 25.25 $\pm$ 0.0 & 10.59 $\pm$ 3.1 &  \underline{2.18 $\pm$ 0.0} & 19.36 $\pm$ 0.0 & 6.81 $\pm$ 0.0 &  7.63 $\pm$ 0.0 & \textbf{2.05 $\pm$ 0.1} \\
    \#3 Water & 29.08 $\pm$ 6.6 &  9.71 $\pm$ 0.0 & 25.37 $\pm$ 0.0 &  7.17 $\pm$ 1.7 &  7.14 $\pm$ 0.0 & 24.06 $\pm$ 0.0 & \underline{4.44 $\pm$ 0.0} &  \textbf{3.55 $\pm$ 0.1} & 6.45 $\pm$ 0.2 \\
    \#4 Road  & 52.93 $\pm$ 8.7 & 50.21 $\pm$ 0.0 & 79.69 $\pm$ 0.0 &  9.41 $\pm$ 0.2 & 34.14 $\pm$ 0.0 & 31.33 $\pm$ 0.0 & \textbf{2.82 $\pm$ 0.0} & 10.66 $\pm$ 0.6 & \underline{5.35 $\pm$ 0.4} \\
    \hline
    Avg. & 35.73 $\pm$ 2.7 & 19.97 $\pm$ 0.0 & 37.99 $\pm$ 0.0 & 12.18 $\pm$ 1.9 & 12.67 $\pm$ 0.0 & 37.84 $\pm$ 0.0 & \underline{4.94 $\pm$ 0.0} &  9.13 $\pm$ 0.2 & \textbf{3.96 $\pm$ 0.1} \\
    \hline
  \end{tabular}}}
  
  \,
  
  {\scriptsize
  \resizebox{\textwidth}{!}{%
  \begin{tabular}{|
      p{1.25cm}||
      >{\centering\arraybackslash}p{1.5cm}
      >{\centering\arraybackslash}p{1.5cm}
      >{\centering\arraybackslash}p{1.7cm}
      >{\centering\arraybackslash}p{1.6cm}
      >{\centering\arraybackslash}p{1.6cm}
      >{\centering\arraybackslash}p{1.6cm}
      >{\centering\arraybackslash}p{1.6cm}
      >{\centering\arraybackslash}p{1.6cm}
      >{\centering\arraybackslash}p{1.7cm}|
    }
    \hline
    \multirow{2}{*}{Endm.} &
    \multicolumn{9}{c|}{\textbf{Abundance Root Mean Square Error (RMSE)} $(\times 10^{-2})$} \\
    \cline{2-10}
    & \textbf{VCA} & \textbf{SISAL} & \textbf{NMF-QMV} & \textbf{Endnet} & \textbf{MiSiCNet} & \textbf{ADMMNet} & \textbf{EDAA} & \textbf{UNMamba} & \textbf{Proposed} \\
    \hline
    \#1 Soil  & 22.08 $\pm$ 2.2 & 22.58 $\pm$ 0.0 & 19.97 $\pm$ 0.0 & 26.27 $\pm$ 1.6 & 21.53 $\pm$ 0.0 & 23.85 $\pm$ 0.0 & \underline{6.96 $\pm$ 0.0} & 10.41 $\pm$ 0.1 &  \textbf{5.44 $\pm$ 0.1} \\
    \#2 Tree  & 12.85 $\pm$ 1.7 &  9.37 $\pm$ 0.0 & 14.55 $\pm$ 0.0 & 25.96 $\pm$ 2.8 &  \textbf{3.35 $\pm$ 0.1} & 17.25 $\pm$ 0.0 & \underline{5.63 $\pm$ 0.0} &  6.79 $\pm$ 0.1 & 7.60 $\pm$ 0.1 \\
    \#3 Water & 11.01 $\pm$ 4.5 & 14.20 $\pm$ 0.0 & 19.81 $\pm$ 0.0 & 40.49 $\pm$ 0.8 &  7.09 $\pm$ 0.0 & 19.94 $\pm$ 0.0 & \textbf{5.35 $\pm$ 0.0} &  \underline{6.38 $\pm$ 0.1} &  9.00 $\pm$ 0.1 \\
    \#4 Road  & 31.02 $\pm$ 4.2 & 26.31 $\pm$ 0.0 & 26.13 $\pm$ 0.0 & 22.29 $\pm$ 1.1 & 24.85 $\pm$ 0.0 & 31.25 $\pm$ 0.0 & \underline{8.73 $\pm$ 0.0} & 9.82 $\pm$ 0.4 &  \textbf{6.42 $\pm$ 0.1} \\
    \hline
    Avg. & 19.24 $\pm$ 1.7 & 18.12 $\pm$ 0.0 & 20.12 $\pm$ 0.0 & 28.75 $\pm$ 1.5 & 14.21 $\pm$ 0.0 & 23.07 $\pm$ 0.0 & \textbf{6.67 $\pm$ 0.0} & 8.35 $\pm$ 0.2 & \underline{7.12 $\pm$ 0.1} \\
    \hline
    \hline
    Time (s) & 0.03 & 0.14 & 7.54 & 453.18 & 112.81 & 52.44 & 27.19 & 22.46 & \textcolor{colorEX1}{369.51} $+$ \textcolor{colorEX2}{0.11} \\
    \hline
  \end{tabular}}}
\end{table*}


\subsection{Experimental setup}

Our method is implemented in Python. We use the unbiased linear SVC solver of Scikit-learn~\cite{pedregosa2011scikit}. 
The simplex \textit{saturation} hyperparameter $s$ is set to
\begin{equation*}
    s = \frac{1}{2\,\text{std}(D')} ,
\end{equation*}
as default value, where $D'$ is the distance matrix after the change of basis. 

\textbf{Semantic segmentation}. 
For Samson, we use a GMM~\cite{li2013hyperspectral} as clustering model fitted on a random $25\%$ data sampling. For Jasper Ridge, we use the blind classification method provided in~\cite{xu2021unified}. For Urban-6, we use the EGFSC-AXBW hyperspectral clustering algorithm~\cite{zhang2025elastic}.

\textbf{Sampling for SVM}. 
To improve performance, we fit SVMs on a random $10\%$ to $30\%$ pixel subset of $Y'$. 
This induces variation in unmixing results from one run to another. We thus report mean$\pm$std over 10 runs per experiment.

\textbf{Baselines}. 
We used the publicly-available codes provided in~\cite{zouaoui2023entropic,ozkan2018endnet,chen2025unmamba} and in the HySUPP Python package~\cite{rasti2024image}, with the same setup and hyperparameters as the ones either used or recommended by the authors of the original papers. 


\begin{table*}[t]
\caption{The SADs and RMSEs (mean$\pm$std over 10 runs per experiment) on the Urban-6 dataset. The best results are shown in bold; the second best are underlined.}
\label{tab:results_urban}
  \centering
  \small
  \setlength{\tabcolsep}{8pt}
  \renewcommand{\arraystretch}{1.2}

  {\scriptsize
  \resizebox{\textwidth}{!}{%
  \begin{tabular}{|
      p{1.25cm}||
      >{\centering\arraybackslash}p{1.5cm}
      >{\centering\arraybackslash}p{1.5cm}
      >{\centering\arraybackslash}p{1.7cm}
      >{\centering\arraybackslash}p{1.6cm}
      >{\centering\arraybackslash}p{1.6cm}
      >{\centering\arraybackslash}p{1.6cm}
      >{\centering\arraybackslash}p{1.6cm}
      >{\centering\arraybackslash}p{1.6cm}
      >{\centering\arraybackslash}p{1.7cm}|
    }
    \hline
    \multirow{2}{*}{Endm.} &
    \multicolumn{9}{c|}{\textbf{Endmember Spectral Angle Distance (SAD)} $(\times 10^{-2})$} \\
    \cline{2-10}
    & \textbf{VCA} & \textbf{SISAL} & \textbf{NMF-QMV} & \textbf{Endnet} & \textbf{MiSiCNet} & \textbf{ADMMNet} & \textbf{EDAA} & \textbf{UNMamba} & \textbf{Proposed} \\
    \hline
    \#1 Asph. & 54.30 $\pm$ $>$9 & 89.21 $\pm$ 0.0 & 42.81 $\pm$ 0.0 &  6.44 $\pm$ 0.9 & 12.98 $\pm$ 0.0 & $>$99 $\pm$ 0.0 &  \underline{4.34 $\pm$ 0.0} & 18.44 $\pm$ 0.2 &  \textbf{3.50 $\pm$ 0.2} \\
    \#2 Grass & 46.65 $\pm$  8.4 & 31.42 $\pm$ 0.0 & 37.27 $\pm$ 0.0 &  \underline{9.46 $\pm$ 1.1} & 18.36 $\pm$ 0.0 & 36.17 $\pm$ 0.0 & 19.74 $\pm$ 0.0 & 16.27 $\pm$ 0.3 &  \textbf{5.08 $\pm$ 0.1} \\
    \#3 Tree  & 34.54 $\pm$  3.3 & 22.06 $\pm$ 0.0 & 39.96 $\pm$ 0.0 & 13.15 $\pm$ 0.6 & 16.62 $\pm$ 0.0 & 26.35 $\pm$ 0.0 & 19.73 $\pm$ 0.0 &  \underline{8.50 $\pm$ 0.0} &  \textbf{7.48 $\pm$ 0.2} \\
    \#4 Roof  & 71.74 $\pm$ $>$9 & 29.53 $\pm$ 0.0 & 85.19 $\pm$ 0.0 & 36.90 $\pm$ 2.5 & 23.53 $\pm$ 0.0 & 40.85 $\pm$ 0.0 & 17.79 $\pm$ 0.0 &  \underline{6.92 $\pm$ 0.1} &  \textbf{6.59 $\pm$ 0.6} \\
    \#5 Metal & 64.28 $\pm$  8.7 & 44.79 $\pm$ 0.0 & 49.18 $\pm$ 0.0 & \textbf{11.11 $\pm$ 1.2} & 15.61 $\pm$ 0.0 & 45.09 $\pm$ 0.0 & 44.70 $\pm$ 0.0 & 22.66 $\pm$ 0.2 & \underline{11.58 $\pm$ 0.8} \\
    \#6 Dirt  & 39.63 $\pm$ $>$9 & 24.72 $\pm$ 0.0 & 49.24 $\pm$ 0.0 & 29.28 $\pm$ 2.6 & 34.06 $\pm$ 0.0 & 24.20 $\pm$ 0.0 & 18.75 $\pm$ 0.0 & \textbf{10.25 $\pm$ 0.4} & \underline{11.17 $\pm$ 0.2} \\
    \hline
    Avg.      & 51.85 $\pm$  1.9 & 40.29 $\pm$ 0.0 & 50.61 $\pm$ 0.0 & 17.72 $\pm$ 0.7 & 20.19 $\pm$ 0.0 & 47.23 $\pm$ 0.0 & 20.84 $\pm$ 0.0 & \underline{13.84 $\pm$ 0.1} & \textbf{7.57 $\pm$ 0.2} \\
    \hline
  \end{tabular}}}
  
  \,

  {\scriptsize
  \resizebox{\textwidth}{!}{%
  \begin{tabular}{|
      p{1.25cm}||
      >{\centering\arraybackslash}p{1.5cm}
      >{\centering\arraybackslash}p{1.5cm}
      >{\centering\arraybackslash}p{1.7cm}
      >{\centering\arraybackslash}p{1.6cm}
      >{\centering\arraybackslash}p{1.6cm}
      >{\centering\arraybackslash}p{1.6cm}
      >{\centering\arraybackslash}p{1.6cm}
      >{\centering\arraybackslash}p{1.6cm}
      >{\centering\arraybackslash}p{1.7cm}|
    }
    \hline
    \multirow{2}{*}{Endm.} &
    \multicolumn{9}{c|}{\textbf{Abundance Root Mean Square Error (RMSE)} $(\times 10^{-2})$} \\
    \cline{2-10}
    & \textbf{VCA} & \textbf{SISAL} & \textbf{NMF-QMV} & \textbf{Endnet} & \textbf{MiSiCNet} & \textbf{ADMMNet} & \textbf{EDAA} & \textbf{UNMamba} & \textbf{Proposed} \\
    \hline
    \#1 Asph. & 29.56 $\pm$ 5.2 & 33.97 $\pm$ 0.0 & 25.40 $\pm$ 0.0 & 30.26 $\pm$ 0.2 & 18.86 $\pm$ 0.0 & 31.16 $\pm$ 0.0 & \underline{15.33 $\pm$ 0.0} & 19.37 $\pm$ 0.2 & \textbf{10.40 $\pm$ 0.3} \\
    \#2 Grass & 35.16 $\pm$ 6.4 & 35.51 $\pm$ 0.0 & 27.97 $\pm$ 0.0 & 39.79 $\pm$ 0.1 & 20.37 $\pm$ 0.0 & 35.25 $\pm$ 0.0 & 43.24 $\pm$ 0.0 & \textbf{14.76 $\pm$ 0.1} & \underline{18.63 $\pm$ 0.2} \\
    \#3 Tree  & 18.07 $\pm$ 1.3 & 16.69 $\pm$ 0.0 & 19.87 $\pm$ 0.0 & 34.50 $\pm$ 0.1 & \underline{12.26 $\pm$ 0.0} & 18.87 $\pm$ 0.0 & 21.63 $\pm$ 0.0 & 12.96 $\pm$ 0.5 & \textbf{11.15 $\pm$ 0.1} \\
    \#4 Roof  & 17.49 $\pm$ 1.6 & 16.39 $\pm$ 0.0 & 17.94 $\pm$ 0.0 & 21.11 $\pm$ 0.1 &  \underline{7.12 $\pm$ 0.0} & 16.56 $\pm$ 0.0 &  \textbf{4.82 $\pm$ 0.0} &  8.23 $\pm$ 0.3 &  8.58 $\pm$ 0.3 \\
    \#5 Metal & 13.17 $\pm$ 2.9 & 18.31 $\pm$ 0.0 & 16.82 $\pm$ 0.0 & 17.17 $\pm$ 0.0 &  \underline{9.11 $\pm$ 0.0} & 12.92 $\pm$ 0.0 & 35.64 $\pm$ 0.0 &  9.60 $\pm$ 0.3 &  \textbf{8.34 $\pm$ 0.4} \\
    \#6 Dirt  & 36.71 $\pm$ 9.4 & 29.95 $\pm$ 0.0 & 25.04 $\pm$ 0.0 & 20.85 $\pm$ 0.1 & 25.00 $\pm$ 0.0 & 33.19 $\pm$ 0.0 & \textbf{13.09 $\pm$ 0.0} & 19.30 $\pm$ 0.5 & \underline{15.38 $\pm$ 0.3} \\
    \hline
    Avg.      & 25.03 $\pm$ 3.6 & 25.14 $\pm$ 0.0 & 22.17 $\pm$ 0.0 & 27.28 $\pm$ 0.1 & 15.45 $\pm$ 0.0 & 24.66 $\pm$ 0.0 & 22.29 $\pm$ 0.0 & \underline{14.04 $\pm$ 0.2} & \textbf{12.08 $\pm$ 0.2} \\
    \hline
    \hline
    Time (s) & 0.32 & 1.79 & 81.98 & 441.74 & 253.32 & 608.35 & 182.85 & 41.19 & \textcolor{colorEX1}{21.02} $+$ \textcolor{colorEX2}{1.15} \\
    \hline
  \end{tabular}}}
\end{table*}


\subsection{Results}

\Cref{tab:results_samson,tab:results_urban,tab:results_jasper_ridge} report the resulting SAD and RMSE for each material (rows), as well as the average value (Avg.). 
The means and standard deviations over 10 runs are given. 
We also give computation times (in seconds). For our method, we separate the clustering model time (left) and our segmentation-to-unmixing model time (right). 
The visualizations are provided in the supplementary materials.

In terms of endmember recovery (SAD), our method achieves the best average performance across all datasets. 
For abundances (RMSE), we obtain the best results on Samson and Urban-6; for Jasper Ridge, although EDAA achieves better results, our RMSE remains competitive, ranking second and close to the top performer. 
For all datasets, our method consistently yields small (near-best) SADs and RMSEs for every endmember, which is not necessarily observed for the other approaches. 
Overall, these results demonstrate its effectiveness, especially in a more challenging scenario (Urban-6), with a relatively high execution speed. 
They confirm that our method, with a appropriately-chosen segmentation or clustering method, can surpass deep and non-deep state-of-the-art approaches. 

\subsection{Impact of input segmentation on unmixing performance}


Since our unmixing method naturally depends on the input segmentation, we assess here how segmentation quality affects unmixing performance.

\Cref{tab:tab1} reports the unmixing results obtained with our segmentation-to-unmixing model over six input classification maps: the ground truth and the outputs of five clustering algorithms, namely k-Means, GMM, Spectral Clustering~\cite{ng2001spectral}, Self-supervised Double-Structure Transformer (SDST)~\cite{10462168}, and Elastic Graph Fusion Subspace Clustering (EGFSC)~\cite{zhang2025elastic}. 
For each input classification map, we report its accuracy (Acc.), together with the abundance RMSE and endmember SAD produced by our unmixing method, averaged over materials.

These results highlight the importance of the clustering model chosen to unmix a dataset, especially in challenging scenarios like Urban-6. 
Overall, EGFSC provides the best performance among the compared methods, both in terms of segmentation accuracy (Acc.) and unmixing quality (SAD and RMSE), except on Samson for which a GMM is associated with better unmixing results despite providing the lowest segmentation accuracy. 
The results, in particular for Samson and Jasper Ridge, also show that the ground-truth segmentation does not necessarily lead to the best unmixing performance: although it provides the correct class labels, the separating hyperplanes learned by the SVM~\cite{pedregosa2011scikit} are determined by margin maximization rather than by the underlying physical mixing process. As a result, the induced decision boundaries may be suboptimal for endmember and abundance estimation.
This may occur when the pixels of one class lie much closer to the true separating hyperplane than those of the opposite class; misclassifying part of them could then help better fit hyperplanes to the ground-truth ones, and thus improve unmixing performance.

We also evaluated our method on ground-truth segmentation maps corrupted by different levels of label noise, using the same experimental protocol as in~\cite{bottenmuller2025improving}. 
\Cref{tab:tab2} shows that our method remains robust to uniformly distributed random label noise. This can be explained by the well-known robustness of SVMs to uniformly misclassified samples, since the solution mainly depends on support vectors, which may remain stable or be only mildly perturbed under such noise. 
Interestingly, introducing a small amount of noise can even improve the predictions, as observed for Jasper Ridge and Urban-6. One possible explanation is that the noise leads to a better balance between classes and relaxes the position of the support vectors. 
However, the previous experiment (see \Cref{tab:tab1}) suggests that the method remains sensitive to the geometry of the class regions in the classification map, since SVMs are particularly influenced by boundary pixels.


\definecolor{cellcolor1}{rgb}{1.00, 1.00, 1.00}

\newcommand{\noiseheader}[2]{%
  \makecell[c]{%
    \begin{tabular}{@{}c@{\hspace{0.4em}}c@{}}
      \raisebox{-0.5\height}{\textbf{#1\%}} &
      \raisebox{-0.5\height}{\includegraphics[width=.080\linewidth]{#2}}
    \end{tabular}%
  }%
}

\begin{table*}[t]
\caption{Evaluation of the method on different clustering algorithms (Avg. values, $\times 10^{-2}$).}
\label{tab:tab1}
  \centering
  \small
  \setlength{\tabcolsep}{8pt}
  \renewcommand{\arraystretch}{1.2}
  {\scriptsize
  \resizebox{\textwidth}{!}{%
    \begin{tabular}{|l||c|c|c|c|c|c|}
      \hline
      \diagbox{Data}{Method} 
      & \textbf{Ground Truth}
      & \textbf{k-Means}
      & \textbf{GMM}
      & \textbf{Spectral Cl.}
      & \textbf{SDST}
      & \textbf{EGFSC} \\
      \hline
      \hline
      \cellcolor{cellcolor1}Samson & 
      $\begin{array}{c}\text{Acc.: }100.00\\\text{SAD: }2.79\\\text{RMSE: }3.43\end{array}$ &
      $\begin{array}{c}\text{Acc.: } 97.28\\\text{SAD: }2.90\\\text{RMSE: }3.63\end{array}$ &
      $\begin{array}{c}\text{Acc.: } 93.03\\\text{SAD: }2.73\\\text{RMSE: }3.52\end{array}$ &
      $\begin{array}{c}\text{Acc.: } 96.03\\\text{SAD: }2.87\\\text{RMSE: }3.56\end{array}$ &
      $\begin{array}{c}\text{Acc.: } 96.04\\\text{SAD: }2.94\\\text{RMSE: }3.68\end{array}$ &
      $\begin{array}{c}\text{Acc.: } 94.54\\\text{SAD: }2.77\\\text{RMSE: }3.55\end{array}$ \\
      \hline
      \cellcolor{cellcolor1}Jasper Ridge & 
      $\begin{array}{c}\text{Acc.: }100.00\\\text{SAD: } 4.41\\\text{RMSE: } 7.87\end{array}$ &
      $\begin{array}{c}\text{Acc.: } 72.73\\\text{SAD: }10.52\\\text{RMSE: }16.42\end{array}$ &
      $\begin{array}{c}\text{Acc.: } 75.39\\\text{SAD: } 8.37\\\text{RMSE: }12.88\end{array}$ &
      $\begin{array}{c}\text{Acc.: } 77.87\\\text{SAD: } 6.75\\\text{RMSE: }12.52\end{array}$ &
      $\begin{array}{c}\text{Acc.: } 83.76\\\text{SAD: } 5.85\\\text{RMSE: } 8.28\end{array}$ &
      $\begin{array}{c}\text{Acc.: } 88.56\\\text{SAD: } 5.55\\\text{RMSE: } 8.08\end{array}$ \\
      \hline
      \cellcolor{cellcolor1}Urban-6 & 
      $\begin{array}{c}\text{Acc.: }100.00\\\text{SAD: } 6.82\\\text{RMSE: }11.07\end{array}$ &
      $\begin{array}{c}\text{Acc.: } 68.09\\\text{SAD: }18.24\\\text{RMSE: }19.39\end{array}$ &
      $\begin{array}{c}\text{Acc.: } 71.73\\\text{SAD: }15.55\\\text{RMSE: }18.70\end{array}$ &
      $\begin{array}{c}\text{Acc.: } 76.65\\\text{SAD: }23.12\\\text{RMSE: }17.82\end{array}$ &
      $\begin{array}{c}\text{Acc.: } 67.89\\\text{SAD: }17.87\\\text{RMSE: }19.72\end{array}$ &
      $\begin{array}{c}\text{Acc.: } 77.37\\\text{SAD: } 7.57\\\text{RMSE: }12.08\end{array}$ \\
      \hline
  \end{tabular}}}
\end{table*}

\begin{table*}[t]
\caption{Evaluation of the method robustness to uniform random noise (Avg. values, $\times 10^{-2}$).}
\label{tab:tab2}
  \centering
  \small
  \setlength{\tabcolsep}{8pt}
  \renewcommand{\arraystretch}{1.2}
  {\scriptsize
  \resizebox{\textwidth}{!}{%
    \begin{tabular}{|l||c|c|c|c|c|c|}
      \hline
      \diagbox{Data}{Noise} 
      & \noiseheader{1}{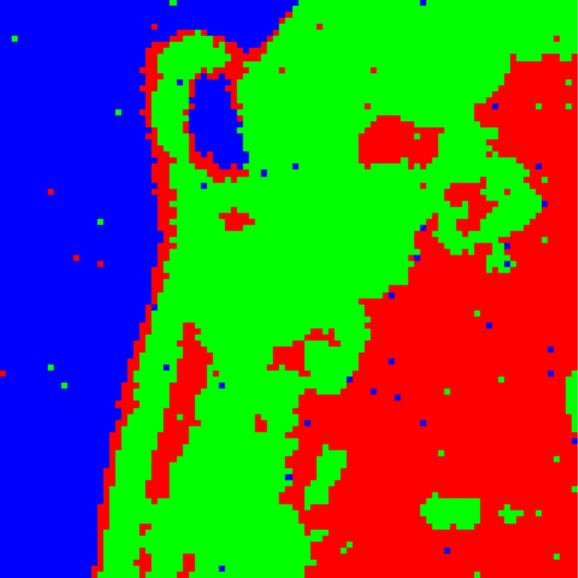}
      & \noiseheader{5}{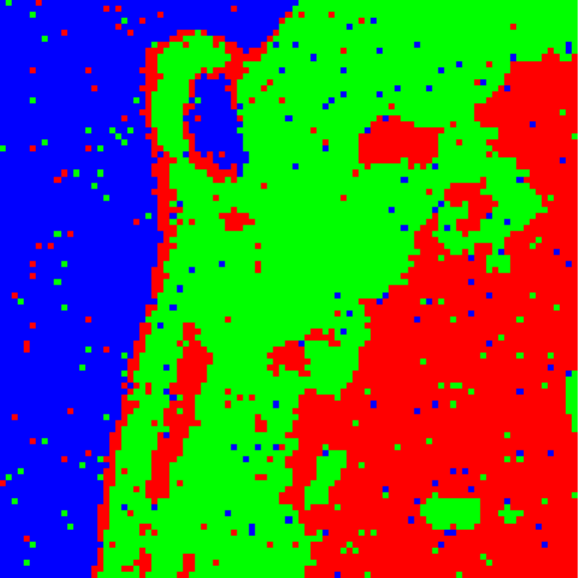}
      & \noiseheader{10}{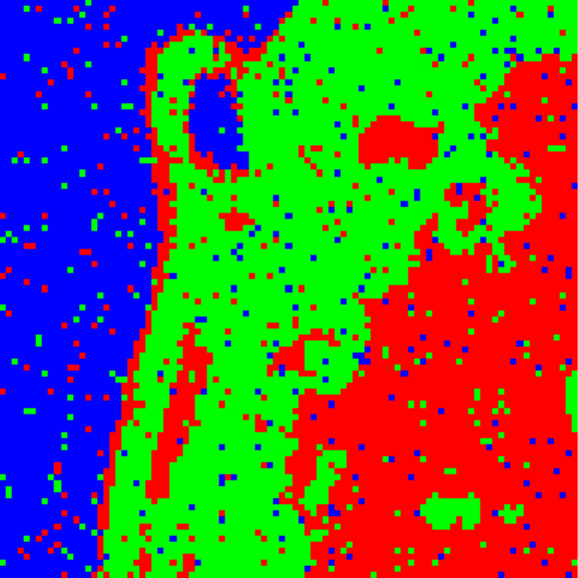}
      & \noiseheader{25}{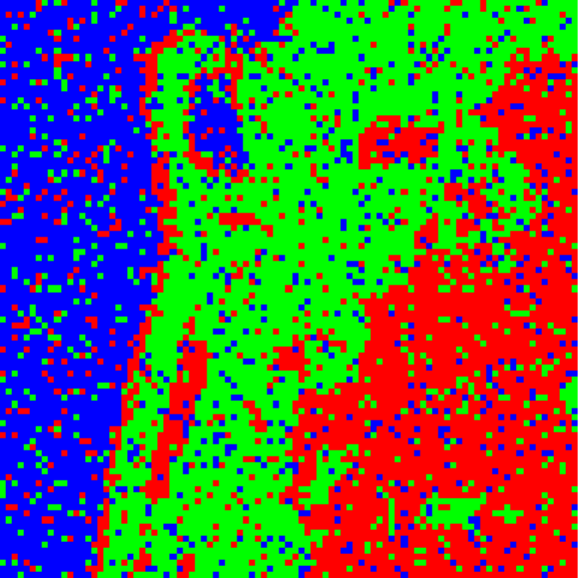}
      & \noiseheader{50}{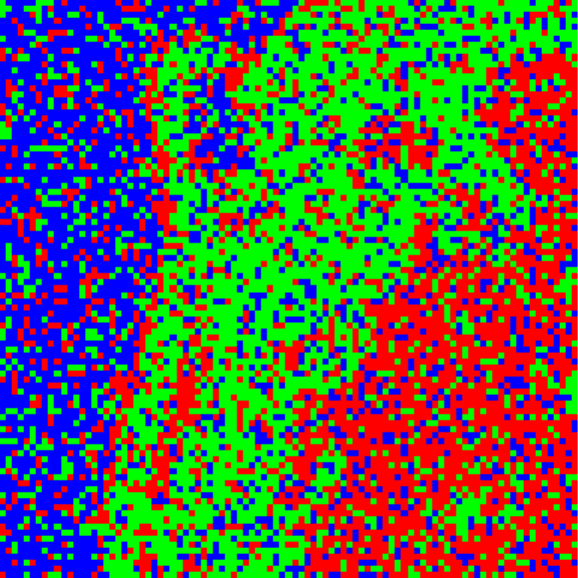}
      & \noiseheader{80}{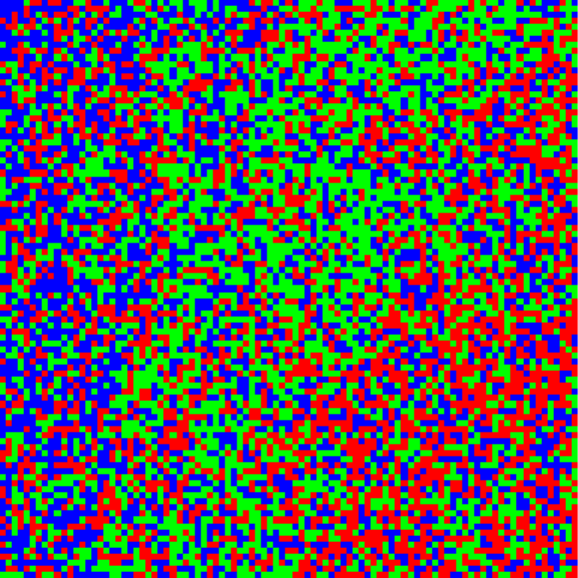} \\
      \hline
      \hline
      \cellcolor{cellcolor1}Samson & 
      $\begin{array}{c}\text{SAD: }2.76\\\text{RMSE: }3.41\end{array}$ &
      $\begin{array}{c}\text{SAD: }2.99\\\text{RMSE: }3.12\end{array}$ &
      $\begin{array}{c}\text{SAD: }3.32\\\text{RMSE: }2.91\end{array}$ &
      $\begin{array}{c}\text{SAD: }4.40\\\text{RMSE: }2.99\end{array}$ &
      $\begin{array}{c}\text{SAD: }6.38\\\text{RMSE: }3.48\end{array}$ &
      $\begin{array}{c}\text{SAD: }5.76\\\text{RMSE: }3.09\end{array}$ \\
      \hline
      \cellcolor{cellcolor1}Jasper Ridge & 
      $\begin{array}{c}\text{SAD: }4.39\\\text{RMSE: }7.78\end{array}$ &
      $\begin{array}{c}\text{SAD: }4.19\\\text{RMSE: }7.23\end{array}$ &
      $\begin{array}{c}\text{SAD: }4.23\\\text{RMSE: }6.97\end{array}$ &
      $\begin{array}{c}\text{SAD: }4.21\\\text{RMSE: }6.52\end{array}$ &
      $\begin{array}{c}\text{SAD: }3.98\\\text{RMSE: }6.07\end{array}$ &
      $\begin{array}{c}\text{SAD: }4.71\\\text{RMSE: }6.76\end{array}$ \\
      \hline
      \cellcolor{cellcolor1}Urban-6 & 
      $\begin{array}{c}\text{SAD: }6.71\\\text{RMSE: }10.92\end{array}$ &
      $\begin{array}{c}\text{SAD: }6.38\\\text{RMSE: } 9.90\end{array}$ &
      $\begin{array}{c}\text{SAD: }6.35\\\text{RMSE: } 9.61\end{array}$ &
      $\begin{array}{c}\text{SAD: }7.02\\\text{RMSE: } 9.58\end{array}$ &
      $\begin{array}{c}\text{SAD: }7.08\\\text{RMSE: } 9.36\end{array}$ &
      $\begin{array}{c}\text{SAD: }9.54\\\text{RMSE: }10.25\end{array}$ \\
      \hline
  \end{tabular}}}
\end{table*}


\section{Conclusion}
\label{sec:conclusion}

We introduced a new blind linear unmixing approach that bridges semantic segmentation and hyperspectral unmixing by inverting the standard pipeline. It estimates endmembers and abundances from any input classification map. 

The core result is a theorem showing that, under the linear mixing model, dominant-material regions partition spectral space into polyhedral cones. We leverage this property by using signed distances to polyhedral cones determined by support vector machines to perform unmixing in the blind setting.

The method is lightweight, easy to deploy, and essentially deterministic. It has a single hyperparameter controlling initial abundance saturation. On Samson, Jasper Ridge and Urban-6, it delivers consistent improvements over eight deep and non-deep baselines, especially for endmember recovery, while remaining competitive in abundance estimation and in computation time. Its effectiveness, however, depends on the quality of the classification map given as input.

Our theorem and proposed segmentation-driven pipeline open new directions for hyperspectral unmixing. We could exploit cone–endmember relations for direct endmember identification or incorporate the polyhedral-cone prior into deep architectures to guide learning and further improve unmixing quality.

\section*{Acknowledgements}
This work was funded by the French Agence Nationale de la Recherche (ANR) under project number ANR 22-CE42-0025.

\clearpage

\appendix
\renewcommand{\thesection}{\Alph{section}}
\titleformat{\section}
  {\normalfont\Large\bfseries}
  {Appendix \thesection.}
  {1em}
  {}
\crefalias{section}{appendix}

\section{Proofs}
\label{app:proofs}

This section provides the proofs of \Cref{lem:convexity,lem:partition} and of \Cref{them:main} from \Cref{sec:method}.

\subsection{Proof of \cref{lem:convexity}}

\begin{proof}
Let $x$ and $x'$ be two elements in $\mathbb{R}^d$. 
By \eqref{eq:linear_combination}, there exist $\lambda, \lambda' \in \mathbb{R}^m$ and $y, y' \in (\text{span}\{M_{:i}\}_{i=1}^m)^\perp$ such that 
\begin{equation*}
    x = \sum_{i=1}^m \lambda_i M_{:i} + y 
    ~~~~~~~\text{ and }~~~~~~~
    x' = \sum_{i=1}^m \lambda_i' M_{:i} + y' ,
\end{equation*}
where $\lambda$ and $\lambda'$ represent the linear coefficient vectors of $x$ and $x'$ over the endmembers $\{M_{:i}\}_{i=1}^m$, and $y$ and $y'$ the components of $x$ and $x'$ in the orthogonal complement $(\text{span}\{M_{:i}\}_{i=1}^m)^\perp$ of $\text{span}\{M_{:i}\}_{i=1}^m$ in $\mathbb{R}^d$, respectively. 

Let us suppose that they both belong to the same dominant-material region $\mathcal{R}_c$, where $c = 1, \ldots, m$, i.e., by Definition \ref{def:region}, $c \in \arg\max_{i = 1}^m \{a_i\}$ and $c \in \arg\max_{i = 1}^m \{a_i'\}$, where $\{a_i\}_{i=1}^m$ and $\{a_i'\}_{i=1}^m$ are the material abundances associated with $x$ and $x'$, respectively. Under the LMM, by linear independence of the endmembers in hyperspectral images, we have $\lambda = (a_1, \ldots, a_m)$ and $\lambda' = (a_1', \ldots, a_m')$. 

Let $x_\rho$ be the vector on the segment $[x,x']$, defined as
\begin{equation*}
    \begin{array}{rrl}
        x_\rho & := & \rho x + (1-\rho) x' \\
         & = & \sum_{i=1}^m (\rho \lambda_i + (1 - \rho) \lambda_i') M_{:i} + \rho y + (1 - \rho) y'
    \end{array}
\end{equation*}
for any $\rho \in [0,1]$.

By linear independence of the endmembers, $x_\rho$ has necessarily for linear coefficient vector $\rho \lambda + (1 - \rho) \lambda'$, which then represents its material abundance vector under the LMM. 
Because the $\arg\max$ is preserved under convex combinations, we have $c \in \arg\max_{i = 1}^m \{\rho \lambda_i + (1 - \rho) \lambda_i'\}$, and so $x_\rho$ belongs to $\mathcal{R}_c$. 
Dominant-material regions in $\mathbb{R}^d$ are therefore convex.
\end{proof}

\subsection{Proof of \cref{lem:partition}}

\begin{proof}
This is a direct consequence to the Hyperplane Separation Theorem (see Theorem 7.3 in~\cite{gallier2011geometric}), pairwise applied to the regions of a convex (finite $m$-) partition of the Euclidean space~{\normalfont{\cite{leon2018spaces}}}.
\end{proof}

\subsection{Proof of \cref{them:main}}

\begin{proof}
Dominant-material regions $\{\mathcal{R}_i\}_{i=1}^m$ form a finite $m$-partition of the spectral space $\mathbb{R}^d$, up to a set of Lebesgue measure zero, i.e., they satisfy the three following properties: 
\begin{itemize}
    \item $\forall i = 1, \ldots, m, \mathcal{R}_i \neq \emptyset$;
    \item $\bigcup_{i=1}^m \mathcal{R}_i = \mathbb{R}^d$;
    \item $\forall i,j = 1, \ldots, m, i \neq j \Rightarrow \mathcal{L}(\mathcal{R}_i \cap \mathcal{R}_j) = 0$,
\end{itemize}
where $\mathcal{L}$ denotes the Lebesgue measure in $\mathbb{R}^d$.

This partitioning property is a direct implication of the $\arg\max$ function properties, which define such regions (see \cref{def:region}). 
By \cref{lem:convexity}, dominant-material regions are convex, therefore forming a convex $m$-partition of $\mathbb{R}^d$ (up to a set of Lebesgue measure zero). 
By \cref{lem:partition}, they thus result in $m$ polyhedral regions of $\mathbb{R}^d$, whose boundaries, by the Hyperplane Separation Theorem, are formed by hyperplanes that pairwise separate the regions. 
These separation hyperplanes are given by the regions in $\mathbb{R}^d$ of dimension $d-1$ where the $\arg\max$ in \eqref{eq:region} is reached for two classes. 

The origin $0_d$ has $0$ for coefficient on every endmember. It thus belongs to every dominant-material region, as the $\arg\max$ is reached for every $0$ component of $0_d$. 
For every pair of classes, $0_d$ then belongs to their separation hyperplane. 
All separation hyperplanes therefore pass through the origin $0_d$, forming, by \cref{def:cone}, polyhedral-cone regions. 
This proves our main Theorem.
\end{proof}

\section{Visualizations}
\label{app:visualizations}

In this section, we provide the visualizations of the three hyperspectral datasets, the unmixing results (endmember signatures and abundance maps) given by the nine tested blind unmixing algorithms on the three datasets, and the classification maps resulting from the five clustering algorithms for hyperspectral image classification that were tested on the three datasets.

\subsection{Datasets}

We provide here the visualizations of the three hyperspectral datasets that are used to test our unmixing algorithm, namely Samson, Jasper Ridge and Urban-6, along with their associated ground-truth abundance maps and endmembers.\\

\noindent\textbf{1. Samson}\\[1pt]

\begin{figure}[!ht]
\centering
    \begin{subfigure}[b]{0.310\linewidth}
        \includegraphics[width=1.0\linewidth]{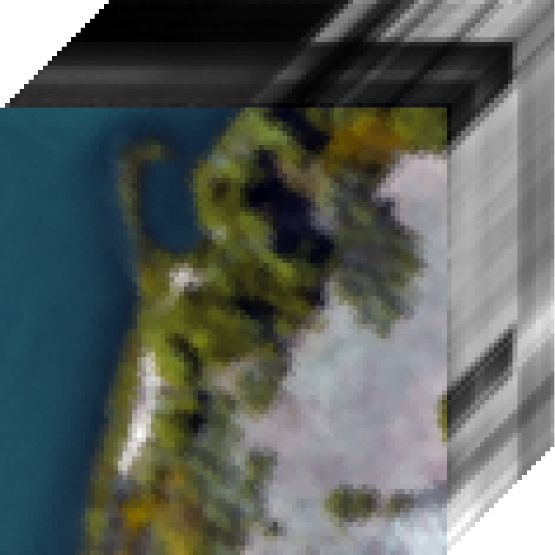}%
        \caption{Hyperspectral datacube.}
        \label{subfig:samson_y}
    \end{subfigure}
    \hfill
    \begin{subfigure}[b]{0.250\linewidth}
        \includegraphics[width=1.0\linewidth]{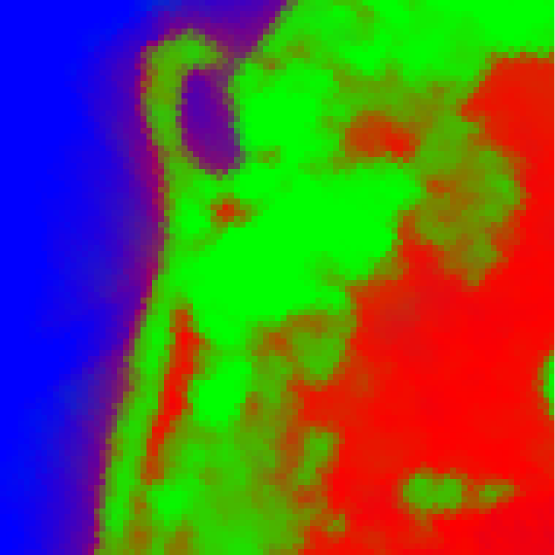}%
        \caption{GT: abundance map.}
        \label{subfig:samson_a}
    \end{subfigure}
    \hfill
    \begin{subfigure}[b]{0.400\linewidth}
        \includegraphics[width=1.0\linewidth]{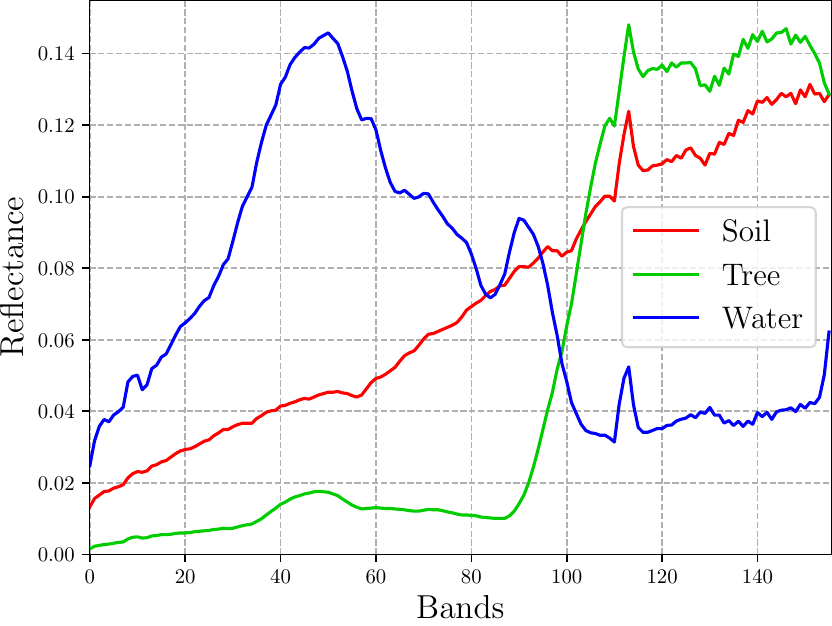}%
        \caption{GT: endmembers.}
        \label{subfig:samson_m}
    \end{subfigure}
    \caption{The Samson dataset (\ref{subfig:samson_y}) with its ground-truth (GT) abundance map (\ref{subfig:samson_a}) and endmembers (\ref{subfig:samson_m}). In subfigures (\ref{subfig:samson_a}) and (\ref{subfig:samson_m}), red is associated with the soil class, green with trees, and blue with water.}
    \label{fig:samson}
\end{figure}

\noindent\textbf{2. Jasper Ridge}\\[1pt]

\begin{figure}[!ht]
\centering
    \begin{subfigure}[b]{0.320\linewidth}
        \includegraphics[width=1.0\linewidth]{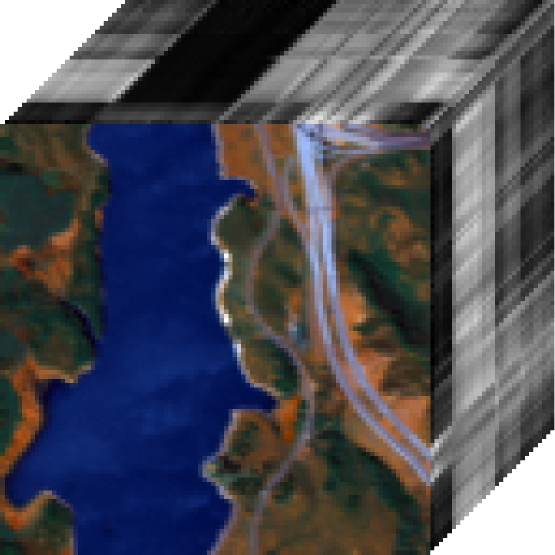}%
        \caption{Hyperspectral datacube.}
        \label{subfig:jasper_y}
    \end{subfigure}
    \hfill
    \begin{subfigure}[b]{0.250\linewidth}
        \includegraphics[width=1.0\linewidth]{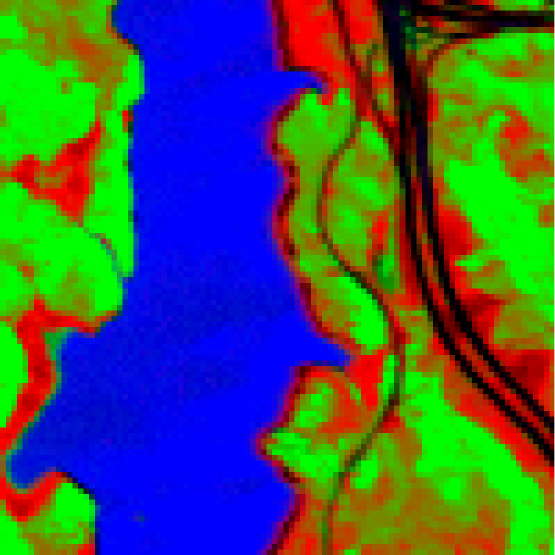}%
        \caption{GT: abundance map.}
        \label{subfig:jasper_a}
    \end{subfigure}
    \hfill
    \begin{subfigure}[b]{0.400\linewidth}
        \includegraphics[width=1.0\linewidth]{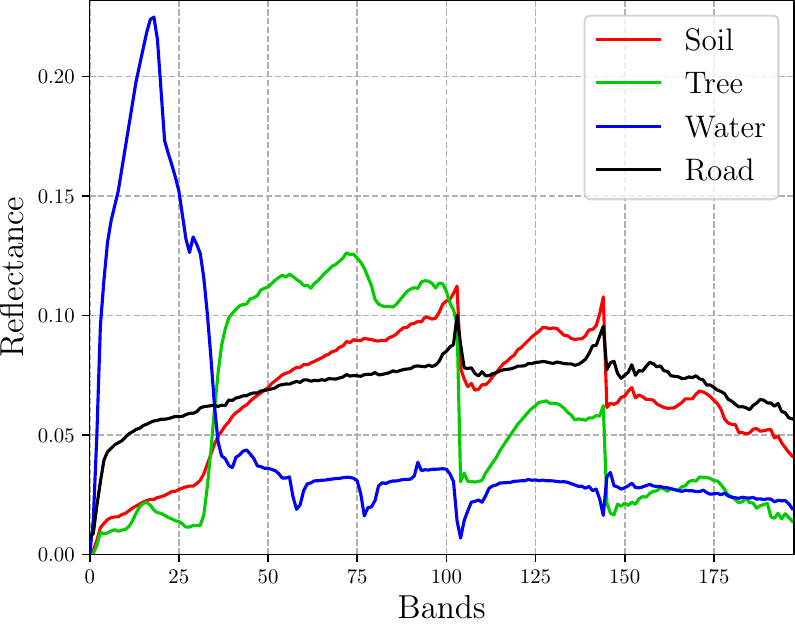}%
        \caption{GT: endmembers.}
        \label{subfig:jasper_m}
    \end{subfigure}
    \caption{The Jasper Ridge dataset (\ref{subfig:jasper_y}) with its ground-truth (GT) abundance map (\ref{subfig:jasper_a}) and endmembers (\ref{subfig:jasper_m}). In subfigures (\ref{subfig:jasper_a}) and (\ref{subfig:jasper_m}), red is associated with the soil class, green with trees, blue with water and black with roads.}
    \label{fig:jasper}
\end{figure}

\clearpage\newpage
\noindent\textbf{3. Urban-6}\\[1pt]

\begin{figure}[!ht]
\centering
    \begin{subfigure}[b]{0.310\linewidth}
        \includegraphics[width=1.0\linewidth]{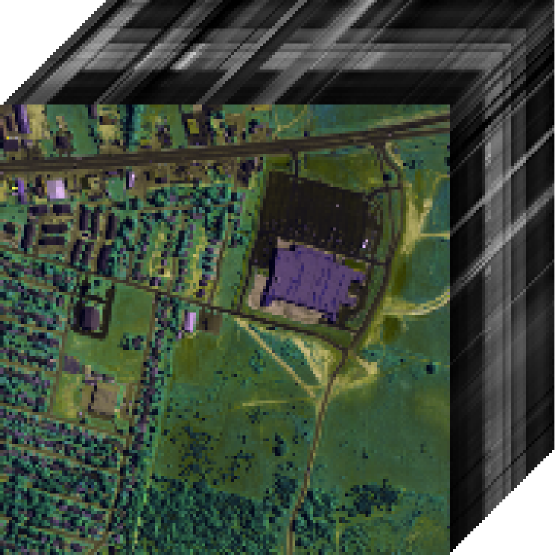}%
        \caption{Hyperspectral datacube.}
        \label{subfig:urban_y}
    \end{subfigure}
    \hfill
    \begin{subfigure}[b]{0.250\linewidth}
        \includegraphics[width=1.0\linewidth]{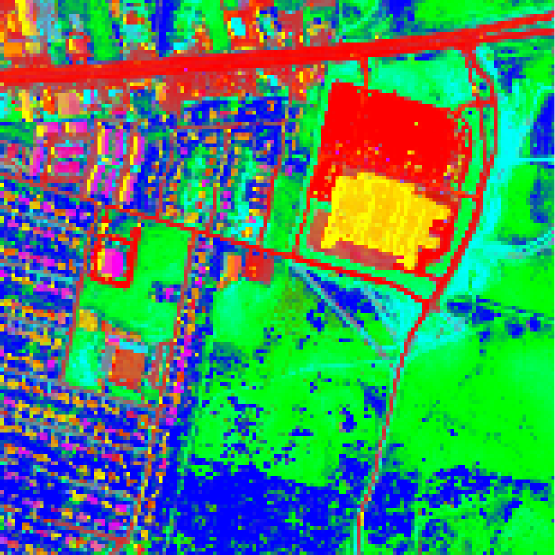}%
        \caption{GT: abundance map.}
        \label{subfig:urban_a}
    \end{subfigure}
    \hfill
    \begin{subfigure}[b]{0.400\linewidth}
        \includegraphics[width=1.0\linewidth]{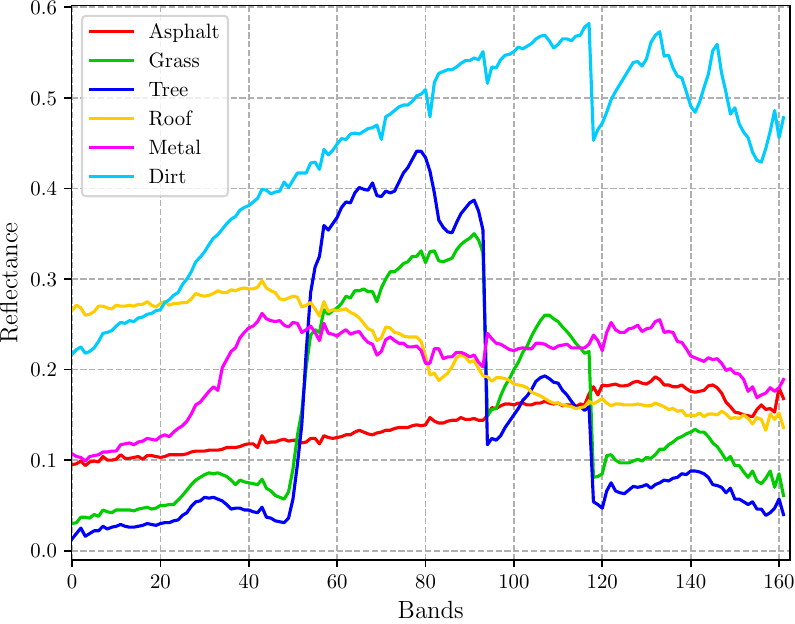}%
        \caption{GT: endmembers.}
        \label{subfig:urban_m}
    \end{subfigure}
    \caption{The Urban-6 dataset (\ref{subfig:urban_y}) with its ground-truth (GT) abundance map (\ref{subfig:urban_a}) and endmembers (\ref{subfig:urban_m}). In subfigures (\ref{subfig:urban_a}) and (\ref{subfig:urban_m}), red is associated with asphalt, green with grass, blue with trees, yellow with roofs, magenta with metal and cyan with dirt.}
    \label{fig:urban}
\end{figure}

\clearpage\newpage
\subsection{Unmixing Results}

We provide here the visualizations of the endmember signatures and abundance maps resulting from the nine blind unmixing algorithms that were tested in this paper, including our proposed polyhedral unmixing algorithm. 
Ground truths are also given as references. 
For each of the three datasets, we provide two figures representing (i) the endmember signatures and (ii) the abundance maps.

The endmember SADs and abundance RMSEs associated with the three datasets are given in \Cref{tab:results_samson,tab:results_jasper_ridge,tab:results_urban}, respectively (see \Cref{sec:experiments}).\\

\noindent\textbf{1. Samson}\\[1pt]

\begin{figure}[!ht]
    \centering
    \includegraphics[width=1.00\linewidth]{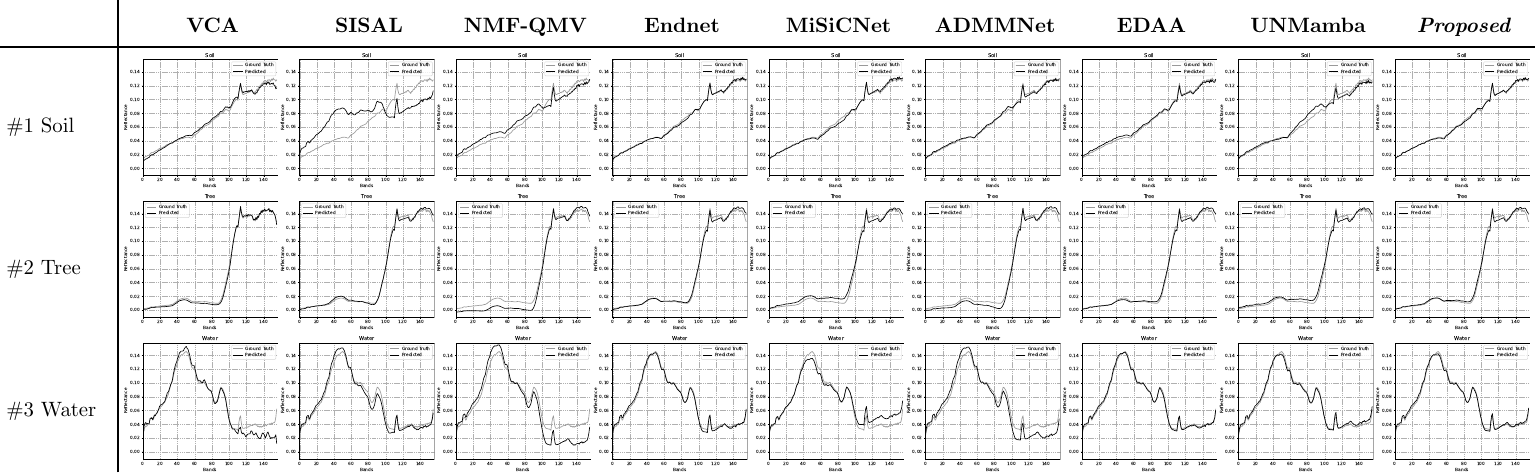}
    \caption{Visualizations of the endmember estimates on \textbf{Samson} given by the nine tested blind unmixing algorithms (columns). The ground-truth endmembers of the three materials (rows) are provided in the associated cells as light-gray line plots.}
    \label{fig:samson_M}
\end{figure}

\begin{figure}[!ht]
    \centering
    \includegraphics[width=1.00\linewidth]{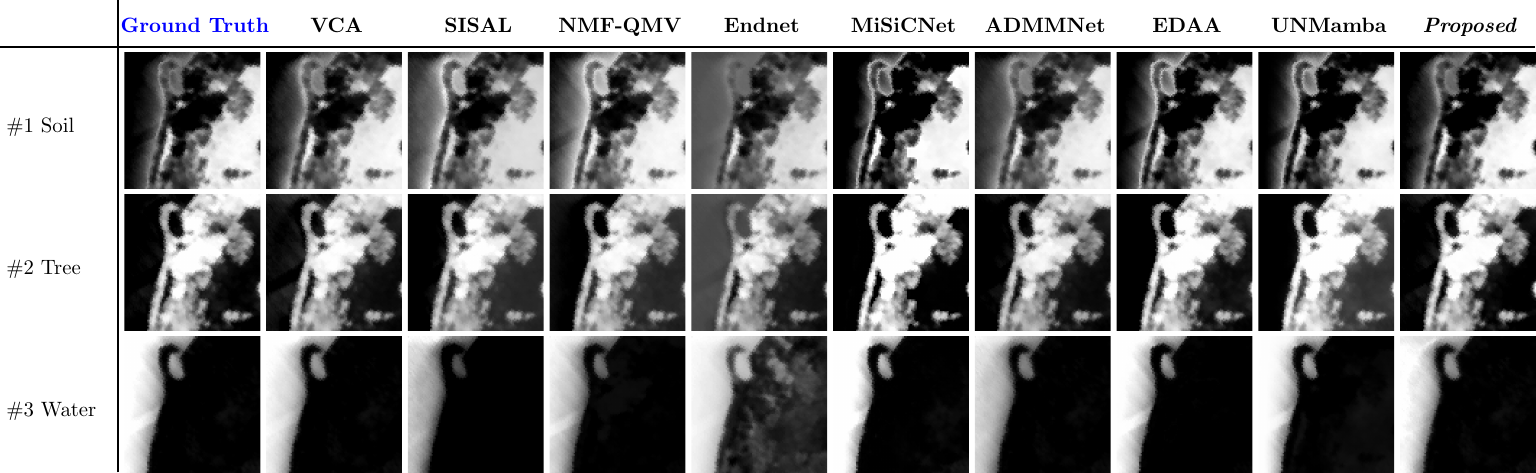}
    \caption{Visualizations of the material abundance estimates on \textbf{Samson} given by the nine tested blind unmixing algorithms (columns). The ground-truth abundance maps of the three materials (rows) are provided in the left column as references.}
    \label{fig:samson_A}
\end{figure}

\clearpage\newpage
\noindent\textbf{2. Jasper Ridge}\\[1pt]

\begin{figure}[!ht]
    \centering
    \includegraphics[width=1.00\linewidth]{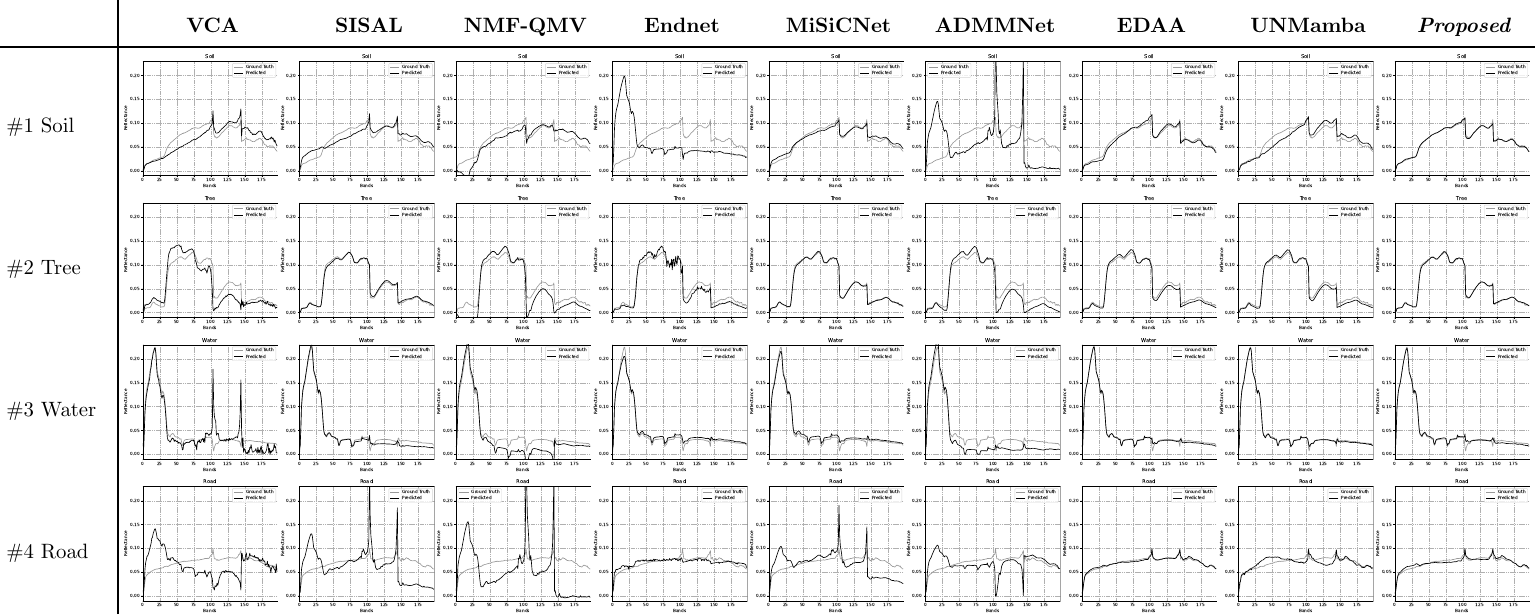}
    \caption{Visualizations of the endmember estimates on \textbf{Jasper Ridge} given by the nine tested blind unmixing algorithms (columns). The ground-truth endmembers of the four materials (rows) are provided in the associated cells as light-gray line plots.}
    \label{fig:jasper_M}
\end{figure}

\begin{figure}[!ht]
    \centering
    \includegraphics[width=1.00\linewidth]{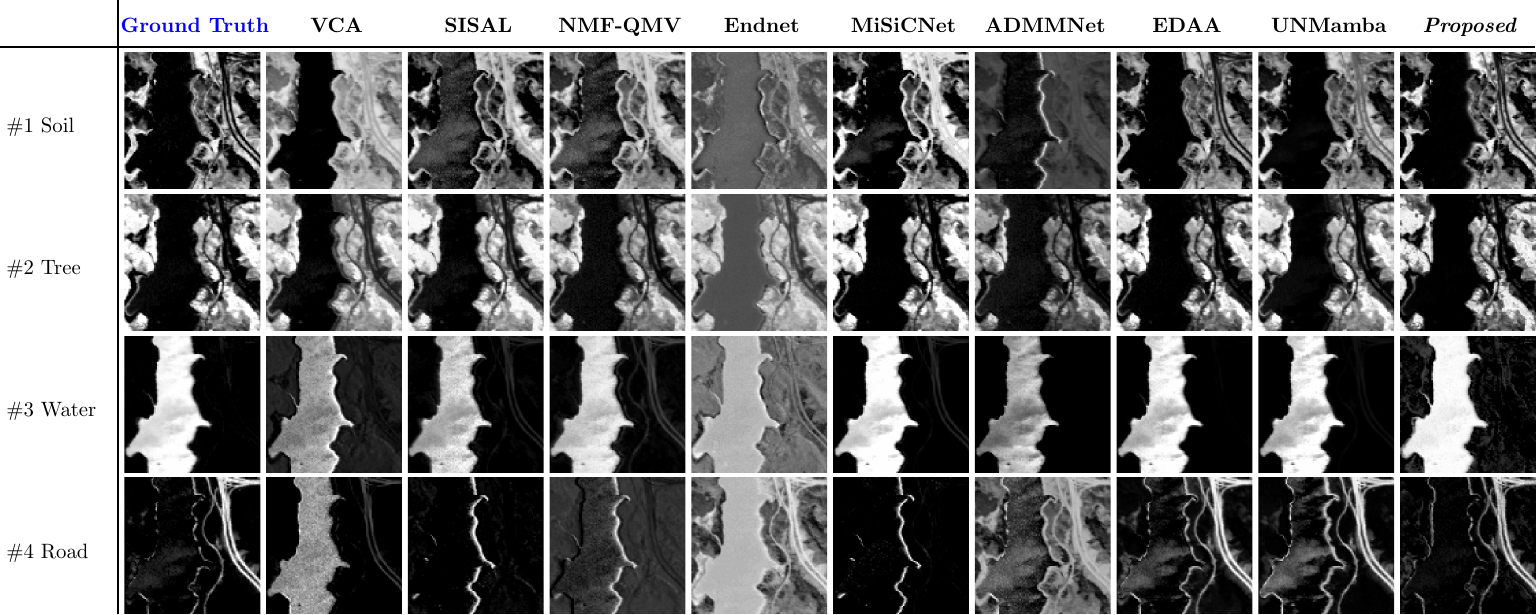}
    \caption{Visualizations of the material abundance estimates on \textbf{Jasper Ridge} given by the nine tested blind unmixing algorithms (columns). The ground-truth abundance maps of the four materials (rows) are provided in the left column as references.}
    \label{fig:jasper_A}
\end{figure}

\clearpage\newpage
\noindent\textbf{3. Urban-6}\\[0pt]

\vspace{18pt}
\noindent\begin{minipage}{\linewidth}
    \centering
    \captionsetup{type=figure}
    \includegraphics[width=1.00\linewidth]{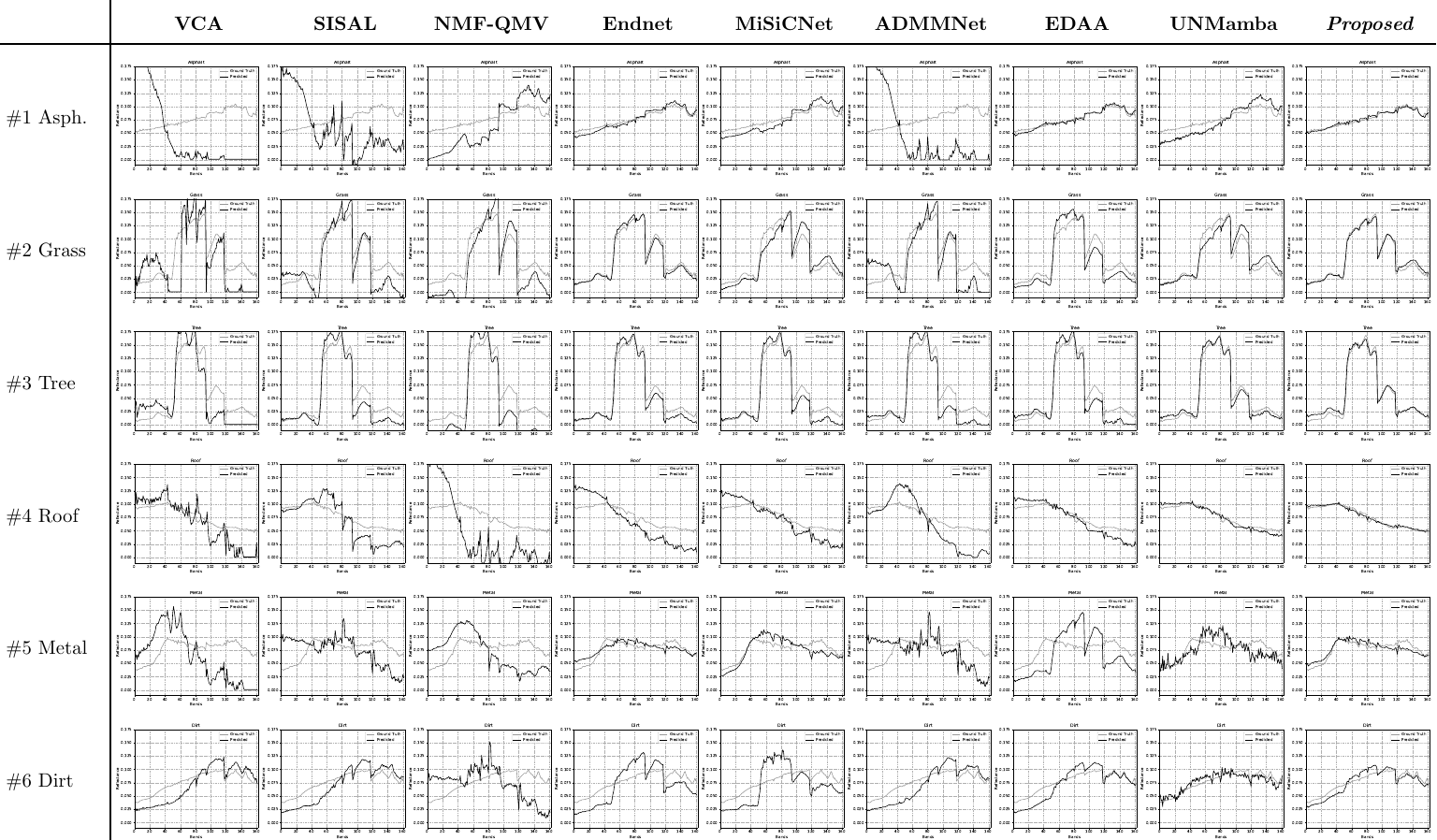}
    \caption{Visualizations of the endmember estimates on \textbf{Urban-6} given by the nine tested blind unmixing algorithms (columns). The ground-truth endmembers of the six materials (rows) are provided in the associated cells as light-gray line plots.}
    \label{fig:urban_M}
    
    \vspace{20pt}

    \centering
    \includegraphics[width=1.00\linewidth]{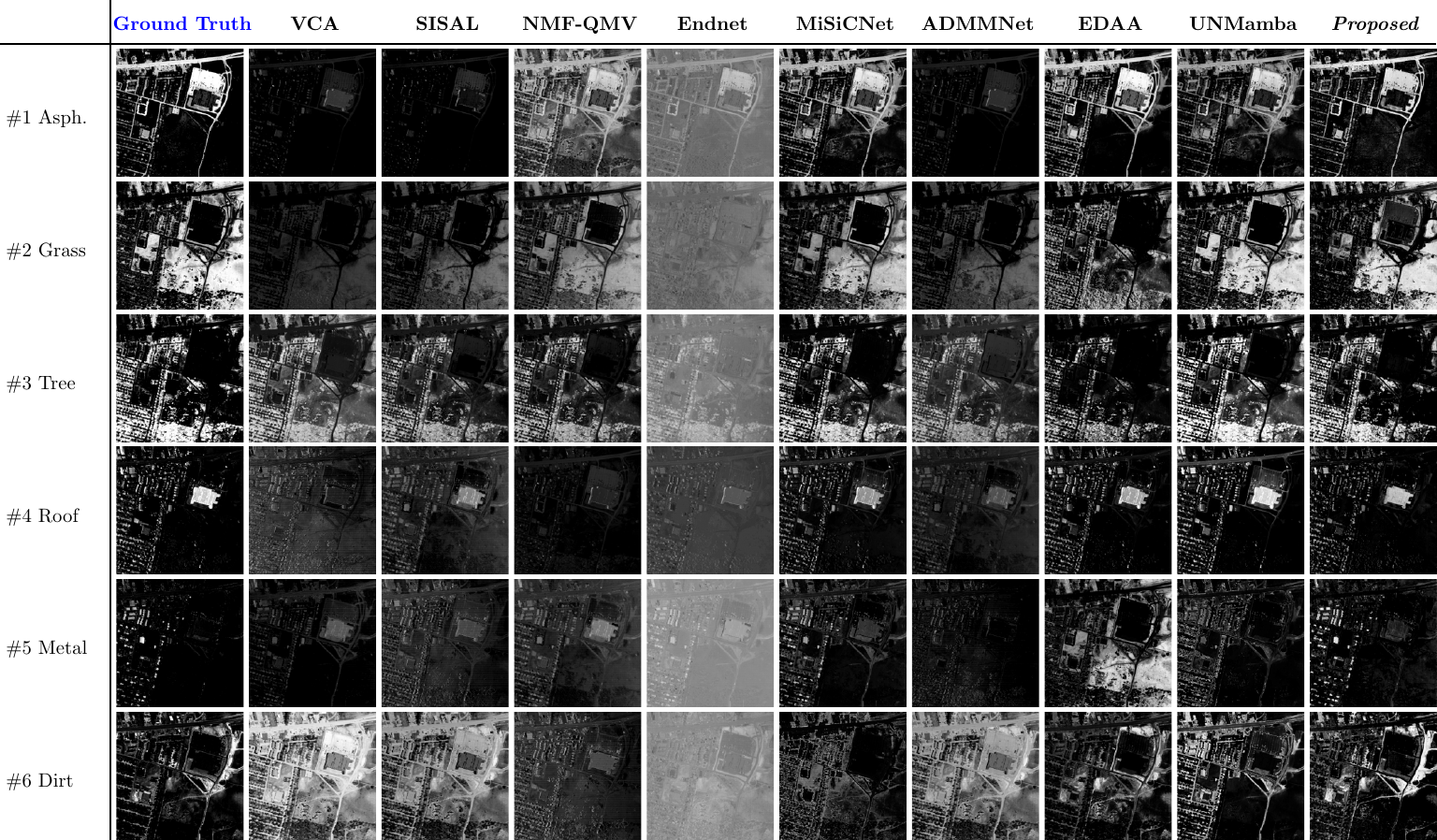}
    \caption{Visualizations of the material abundance estimates on \textbf{Urban-6} given by the nine tested blind unmixing algorithms (columns). The ground-truth abundance maps of the six materials (rows) are provided in the left column as references.}
    \label{fig:urban_A}
\end{minipage}

\subsection{Clustering Results}

We finally provide the classification maps resulting from the five clustering algorithms tested in this paper over the three hyperspectral datasets. The accuracy values between predicted and ground-truth classification maps, along with the estimated endmember SADs and abundance RMSEs resulting from the application of our polyhedral unmixing algorithm over the associated segmentations, are provided in \Cref{tab:tab1} (see \Cref{sec:experiments}).

\begin{figure}[!ht]
    \centering
    \includegraphics[width=1.00\linewidth]{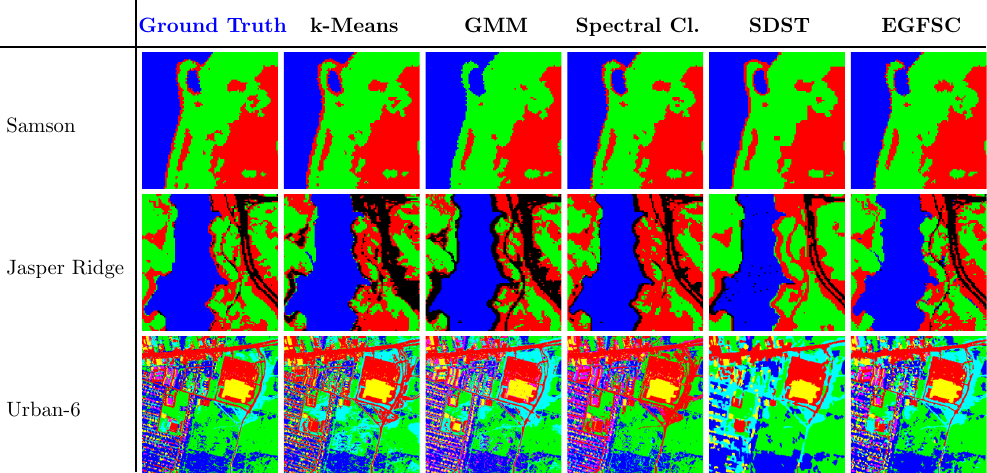}
    \caption{Visualizations of the classification maps given by the five tested clustering algorithms (columns) over the three hyperspectral datasets (rows). The ground-truth classification maps are provided in the left column as references.}
    \label{fig:clustering}
\end{figure}

\clearpage\newpage

%
%
\bibliographystyle{splncs04}
\bibliography{references}

\end{document}